\documentclass{article}
\usepackage{graphicx,amsfonts,amsmath,mathrsfs,amssymb,amsthm,url,color,bbm,subfigure}
\usepackage{fancyhdr,indentfirst,bm,enumitem,natbib}
  \usepackage[colorlinks=true,citecolor=blue]{hyperref}
  \usepackage{booktabs}
         
 \def\oo{\infty}                   \def\d{\,\mathrm{d}}
  \def\lm{\lambda}                  
                    
    \def\dfrac{\displaystyle \frac}

   \newcommand{\E}{\mathbb{E}}
    \newcommand{\R}{\mathbb{R}}
     
      \renewcommand{\P}{\mathbb{P}}
       \newcommand{\cR}{\mathcal{R}}
       
      \newcommand{\VaR}{\operatorname{VaR}}
      \newcommand{\ES}{\operatorname{ES}}
\renewcommand{\(}{\left (}
 \renewcommand{\)}{\right )}
  \renewcommand{\[}{\left [}
   \renewcommand{\]}{\right ]}
\newtheorem{theorem}{Theorem}[section]
 
  \newtheorem{lemma}[theorem]{Lemma}
   \newtheorem{proposition}[theorem]{Proposition}
    \newtheorem{definition}{Definition}[section]
     \newtheorem{example}[theorem]{Example}

\numberwithin{equation}{section}
 \numberwithin{theorem}{section}

\renewcommand{\cite}{\citet}
 \allowdisplaybreaks
  
\def\blue{\color{blue}}

\begin{document}

\baselineskip=18pt

\title{Norms Based on Generalized Expected-Shortfalls and Applications}

\author{ Shuyu Gong\qquad Taizhong Hu\qquad Zhenfeng Zou\footnote{Corresponding author.\newline
 \hbox{\quad\ } E-mail addresses: {\blue shuyu@mail.ustc.edu.cn} (S. Gong), {\blue thu@ustc.edu.cn} (T. Hu), {\blue zfzou@ustc.edu.cn} (Z. Zou) }  \\[10pt]
    Department of Statistics and Finance, School of Management,\\
        University of Science and Technology of China,\\
           Hefei, Anhui 230026, China
     }

\date{July, 2025}

\maketitle

\begin{abstract}

This paper proposes a novel class of generalized Expected-Shortfall (ES) norms constructed via distortion risk measures, establishing a unified analytical framework for risk quantification. The proposed norms extend conventional ES methodology by incorporating flexible distortion functions. Specifically, we develop the mathematical duality theory for generalized-ES norms to support portfolio optimization tasks, while demonstrating their practical utility through projection problem solutions. The generalized-ES norms are also applied to detect anomalies of financial time series data.
\medskip		

		

\end{abstract}

\section{Introduction}

\subsection{Scaled ES-norm}

The \emph{Value-at-Risk} (VaR) and \emph{Expected Shortfall} (ES) are two popular risk measures, which are widely used in banking and insurance. The ES is also termed \emph{Average Value-at-Risk} (AVaR), \emph{Tail Value-at-Risk} (TVaR) and \emph{Conditional Value-at-Risk} (CVaR) in different contexts. For a random variable $X$ with distribution function $F_X$, the VaR of $X$ at confidence level $\alpha$ is defined by
$$
   \VaR_\alpha(X):= F_X^{-1}(\alpha) = \inf\{x: F_X(x)\ge \alpha\},\quad \alpha\in (0,1],
$$
with $\VaR_0(X)=F_X^{-1}(0)={\rm essinf}(X)$ and $\VaR_1(X)=F_X^{-1}(1) ={\rm esssup}(X)$. For $X$ with finite mean, the ES of $X$ at confidence level $\alpha\in [0,1)$ is defined by
\begin{equation*}
   \ES_\alpha(X)=\frac {1}{1-\alpha} \int^1_\alpha {\rm VaR}_s(X) \d s
\end{equation*}
with $\ES_0(X)=\E [X]$ and $\ES_1(X)={\rm esssup}(X)$. It is known from \cite{RU02} that ES has the following representation:
\begin{equation}
 \label{eq-221030-2}
   \ES_\alpha(X)=\min_{t\in\R}\left \{ t+ \frac {1}{1-\alpha} \E [(X-t)_+]\right \},\quad \alpha\in [0,1),
\end{equation}
where $x_+=\max\{x, 0\}$ for $x\in\R$. If $X$ is a discrete random variable distributed on $\{x_1, \ldots, x_n\}$ such that $\P(X=x_i)=p_i$ with $\sum_{i=1}^n p_i=1$, from \eqref{eq-221030-2}, it follows that
\begin{equation}
 \label{eq-250401-1}
    \operatorname{ES}_\alpha(X)=\min_{t\in\R}\left\{t+\frac{1}{1-\alpha} \sum_{i=1}^n p_i (x_i-t)_+\right\}.
\end{equation}

\cite{PU14} introduced two families of scaled and non-scaled ES-norms on $\R^n$ based on the ES risk measure. The scaled ES-norm of a vector $\bm x =(x_1,\ldots,x_n)\in\R^n$ at level $\alpha\in [0,1]$ is denoted by  $\langle\!\langle\bm x\rangle\!\rangle_{\alpha}^S$, and the non-scaled ES-norm of $\bm x$ at level $\alpha\in [0,1]$ is defined to be  $\langle\!\langle\bm x\rangle\!\rangle_{\alpha}= n(1-\alpha) \langle\!\langle\bm x\rangle\!\rangle_{\alpha}^S$. For $\bm x\in\R^n$, $\langle\!\langle\bm x\rangle\!\rangle_{\alpha}^S$ is defined to be the value of $\ES_\alpha(|X|)$, where $X$ is uniformly distributed on $\{x_1, \ldots, x_n\}$. That is,
\begin{equation}
  \label{eq-250401-2}
 \begin{aligned}
     \langle\!\langle\bm x\rangle\!\rangle_\alpha^S & =\min_{t\in\R}\left\{ t+\frac{1}{n(1-\alpha)}
                 \sum_{i=1}^n \(|x_i|-t\)_+\right\},\quad \alpha\in [0,1), \\
     \langle\!\langle\bm x\rangle\!\rangle_1^S & = \max _{1\le i\le n} |x_i|.
\end{aligned}
\end{equation}
In the sequel, let us order the absolute values of components of a vector $\bm x \in \R^n$ as $|x|_{(1)} \le |x|_{(2)} \le \cdots \le |x|_{(n)}$. In \cite{PU14}, another equivalent definition of the scaled ES-norm is also given as follows: for $\bm x\in \R^n$ and $\alpha_j=j/n$,
\begin{equation}
   \label{eq-250401-4}
    \langle\!\langle\bm x\rangle\!\rangle_{\alpha_j}^S = \frac{1}{n-j}\sum_{i=j+1}^n  |x|_{(i)},\quad j=0, \ldots, n-1,
\end{equation}
and for $\alpha$ such that $\alpha_j<\alpha<\alpha_{j+1}$, $j=0,\ldots,n-2$,
\begin{equation}
   \label{eq-250401-5}
   \langle\!\langle \bm x \rangle\!\rangle_{\alpha}^S = \mu \langle\!\langle \bm x \rangle\!\rangle_{\alpha_j}^S + (1-\mu)\langle\!\langle \bm x \rangle\!\rangle_{\alpha_{j+1}}^S,
\end{equation}
where $\mu\in(0,1)$ such that
$$
    \frac{1}{1-\alpha} =\frac{\mu}{1-\alpha_j}+\frac{1-\mu}{1-\alpha_{j+1}}.
$$
For $\alpha$ such that $(n-1)/n <\alpha\le 1$, 
$
   \langle\!\langle\bm x \rangle\!\rangle_\alpha^S =|x|_{(n)}.
$
This definition is inspired in a more intuitive way since $\ES_\alpha$ is actually the tail average of the concerned distribution beyond its $\alpha$-quantile. Moreover, the scaled $L_1^S$-norm is defined to be
$$
    ||\bm x ||_1^S=\frac{1}{n}\sum_{i=1}^n |x_i|.
$$
Thus, the scaled $L_1^S$-norm and the well-known $L_\oo$-norm are limiting cases of the family of the scaled ES-norms $\langle\!\langle\cdot\rangle\!\rangle_{\alpha}^S$ as $\alpha$ goes to $0$ and $1$, respectively.

ES-norm can be efficiently used in optimization problem. \cite{PU14} proved that the ES-norm is equivalent to the D-norm introduced in \cite{BPS04} for robust optimization.

\subsection{Motivation and contributions}

Norms have evolved into a fundamental tool across diverse mathematical disciplines and their applications, spanning probability theory, statistical analysis, pattern recognition, clustering algorithms, signal processing, machine learning, and numerous other scientific domains (see, for example, \cite{DD16,ZF14,ZF15,KLBP20,MMC21}). The rich diversity of norm formulations stems from their wide-ranging applications, necessitating careful selection of context-appropriate norms.

$\ES_\alpha$ quantifies the average loss beyond a given quantile level $\alpha$, thereby capturing tail risks. However, traditional ES-norms assume uniform weighting of losses, limiting their ability to reflect complex risk profiles observed in real-world markets. To address this limitation, we introduce the generalized-ES norm, which incorporates a distortion function to weight different parts of the loss distribution. This framework enables a more flexible and nuanced risk assessment, capturing a variety of risk characteristics through diverse distortion functions.

Our work makes several contributions to both the theoretical and practical aspects of financial risk management. The main contributions of this paper are as follows. We establish necessary and sufficient conditions for the generalized-ES to constitute a proper norm (Proposition \ref{pr-250707}), and develop dual representations that facilitate computational implementation (Proposition \ref{pr-250510}). The generalized-ES norm demonstrates robust anomaly detection capabilities when applied to high-frequency financial time series characterized by heavy-tailed distributions (Subsection \ref{sect-5.2}).

\subsection{Outline of the paper}

The rest of this paper is organized as follows. In Section \ref{sect-2}, we introduce the scaled generalized-ES norm and provide its alternative representations. We also discuss the relationship between distortion functions and the convexity of unit disks, providing examples of commonly used distortion functions. Section \ref{sect-3} extends the framework to non-scaled cases, and discusses its properties. Section \ref{sect-4} develops the dual norm theory for the scaled generalized-ES norm, which is crucial for optimization applications. We establish the connection between the primal and dual representations, and their computational implications. Section \ref{sect-5} demonstrates practical applications through two case studies: (1) solving projection problems, and (2) detecting anomalies in financial time series data. These applications highlight the framework's versatility.

\section{Scaled generalized-ES norm}
\label{sect-2}

We first recall the class of \emph{distortion risk measures} that contains VaR and ES. A distortion risk measure is defined via the Choquet integral,
\begin{equation*}
  \rho_h(X)=\int^\oo_0 h(\P(X>x))\d x -\int^0_{-\oo} [1-h(\P(X>x))] \d x,
\end{equation*}
where $X$ is a random variable in the domain of $\rho_h$, that is, at least one of the two integrals is finite.
The function $h: [0,1]\to [0,1]$ is called a \emph{distortion function}, i.e., $h$ is increasing with $h(0)=0$ and $h(1)=1$. Throughout, ``increasing'' and ``decreasing'' are used to mean ``non-decreasing'' and ``non-increasing'', respectively. For a distortion function $g$, denote its dual distortion function by ${\bar g}(t)=1-g(1-t)$, and define $H_g[X]=\rho_{\bar g}(X)$. For more properties on distortion risk measures, see \cite{DKLT12} and \cite{WWW20}.

The generalized-ES, introduced by \cite{Pic24} and carefully studied by \cite{GZGH24}, is defined by replacing the expectation $\E [\cdot]$ in \eqref{eq-221030-2} by a distortion risk measure $H_g[\cdot]$, that is,
\begin{equation}
   \label{eq-2024-4}
    \cR_\alpha (X) = \inf_{t\in\R}\left\{t+\frac{1}{1-\alpha} H_g\[(X-t)_+\] \right\},
\end{equation}
where $g$ is a distortion function and $\alpha\in [0,1)$ is a confidence level.

In view of the generalized-ES, we can define the following generalized-ES norm, which also has two variations: scaled and non-scaled. For $\bm x=(x_1, \ldots, x_n)\in\R^n$, let $X$ be a random variable uniformly distributed on $\{x_1,\ldots,x_n\}$.
The scaled generalized-ES norm of $\bm x\in\R^n$, denote by  $\langle\!\langle \bm x\rangle\!\rangle_{\alpha}^{S,g}$, is defined by
\begin{equation}
   \label{eq-20250712}
  \langle\!\langle \bm x\rangle\!\rangle_{\alpha}^{S,g} =\cR^g_\alpha(|X|),\quad \alpha\in [0,1).
\end{equation}
The non-scaled generalized-ES norm is introduced in Section \ref{sect-3}. Note that $H_g[X]$ is a weighted summation, $H_g[X]=\sum_{i=1}^n c_i\, x_{(i)}$, where $c_i= g(i/n)-g((i-1)/n)$. Therefore, we have

\begin{definition}
 \label{def-GESnorm}
For $\bm x\in\R^n$, its scaled generalized-ES norm based on $\cR^g_\alpha$ is defined as follows:
\begin{equation}
  \label{eq-250401-3}
   \begin{aligned}
      \langle\!\langle\bm x\rangle\!\rangle_\alpha^{S,g} & =\min_{t\in\R}\left\{ t+\frac{1}{1-\alpha} \sum_{i=1}^n c_i ( |x|_{(i)}-t )_+\right\},\quad \alpha\in [0,1), \\
     \langle\!\langle \bm x\rangle\!\rangle_1^{S,g} & = \max_{1\le i\le n} |x_i| = |x|_{(n)},
   \end{aligned}
\end{equation}
where $c_i=g(i/n)- g((i-1)/n)$ for $i=1, \ldots, n$.
\end{definition}

Particularly, $\langle\!\langle \cdot \rangle\!\rangle_1^{S,g}$ coincides with the $L^\oo$ norm. As for $\alpha=0$,
\begin{equation}
 \label{eq-250401-6}
    \langle\!\langle \bm x \rangle\!\rangle_0^{S,g} =\min_{t\in\R} \left\{t+\sum_{i=1}^n c_i (|x|_{(i)}-t)_+\right\} = \sum_{i=1}^n c_i |x|_{(i)}.
\end{equation}
When the distortion function $g(t)=t$ for $t\in [0,1]$, we have $c_i=1/n$ for $i=1,\ldots,n$. Thus the scaled generalized-ES norm reduces to the scaled ES-norm. From the definition, it is easy to see that the scaled generalized-ES norm $\langle\!\langle \bm x\rangle\!\rangle_\alpha^{S,g}$ is increasing in $\alpha$ for fixed $\bm x\in\R^n$.

\begin{proposition}
 \label{pr-250707} 
If $g$ is a strictly increasing and convex distortion function, then $\langle\!\langle\cdot\rangle\!\rangle_\alpha^{S,g}$ is a norm for $\alpha\in [0,1]$, that is, it satisfies the following three properties:
\begin{enumerate}[label={\rm (\roman*)}]
  \item $\langle\!\langle\lm\cdot\bm x\rangle\!\rangle_\alpha^{S,g} = |\lm| \cdot\langle\!\langle \bm x\rangle\!\rangle_\alpha^{S,g}$ for $\lm\in\R$ and $\bm x\in\R^n$;

  \item $\langle\!\langle\bm x +\bm y \rangle\!\rangle_\alpha^{S,g}\le  \langle\!\langle\bm x \rangle\!\rangle_\alpha^{S,g} +\langle\!\langle\bm y \rangle\!\rangle_\alpha^{S,g}$ for $\bm x, \bm y\in\R^n$;

  \item $\langle\!\langle\bm x\rangle\!\rangle_\alpha^{S,g}=0$ if and only if $\bm x=\bm 0$.
\end{enumerate}
\end{proposition}

\begin{proof}
(i)\ For $\alpha\in [0,1)$,
   \begin{align*}
      \langle\!\langle \lm \cdot \bm x \rangle\!\rangle_\alpha^{S,g}
       & =\min_{t\in\R}\left\{ t+\frac{1}{1-\alpha} \sum_{i=1}^n c_i \( |\lm x|_{(i)} -t\)_+\right\} \\
       & =\min_{t\in\R}\left\{ |\lm| t+\frac{1}{1-\alpha} \sum_{i=1}^n c_i \( |\lm x|_{(i)} -|\lm| t\)_+\right\} \\
        & =|\lm| \min_{t\in\R}\left\{ t+\frac{1}{1-\alpha} \sum_{i=1}^n c_i \( |x|_{(i)} -t\)_+\right\}
         = |\lm| \langle\!\langle \bm x \rangle\!\rangle_\alpha^{S,g}.
  \end{align*}
  For $\alpha=1$,
  $$
     \langle\!\langle \lm\cdot \bm x\rangle\!\rangle_1^{S,g} = \max_{1\le i\le n} |\lm x|_{(i)}
       = |\lm|\cdot |x|_{(n)} = |\lm| \langle\!\langle \bm x \rangle\!\rangle_1^{S,g}.
  $$

  (ii)\ For $\alpha=1$, it is trivial. We consider $\alpha\in[0,1)$. For $\bm x, \bm y\in \R^n$, define two random variables $X=x_I$ and $Y=y_I$, where $I$ follows the uniform distributions on $\{1, \ldots, n\}$. By Proposition 2.1 in \cite{GZGH24}, $\cR_\alpha$ possesses monotonicity and subadditivity since $g$ is increasing and convex. Therefore,
  \begin{align*}
    \langle\!\langle \bm x +\bm y \rangle\!\rangle_\alpha^{S,g}
      & = \cR^g_{\alpha} (|X+Y|) \le \cR^g_{\alpha} (|X| + |Y|) 
       \le \cR^g_{\alpha} (|X|) + \cR^g_{\alpha} (|Y|)
      = \langle\!\langle \bm x \rangle\!\rangle_\alpha^{S,g}
        + \langle\!\langle \bm y \rangle\!\rangle_\alpha^{S,g}.
\end{align*}

(iii)\ It suffices to prove the necessity. Since $\langle\!\langle\bm x\rangle\!\rangle_{\alpha}^{S,g}$ is increasing in $\alpha$, we have
  $$
     0=\langle\!\langle\bm x\rangle\!\rangle_{\alpha}^{S,g} \ge \langle\!\langle \bm x\rangle\!\rangle_0^{S,g} = \sum_{i=1}^n c_i |x|_{(i)}=0.
  $$
  Consequently, $x_i=0$ for $i=1,\ldots,n$ since $g$ is strictly increasing.
\end{proof}

Though Proposition \ref{pr-250707} establishes that $g$ being a strictly increasing and convex function is a necessary and sufficient condition for $\langle\!\langle\cdot \rangle\!\rangle_{\alpha}^{S,g}$ to be a norm, many results discussed subsequently do not require $g$ to be convex unless explicitly stated. 

By Theorem 3.1 (\emph{Representation theorem of generalized-ES}) in \cite{GZGH24}, we give an equivalent definition of the scaled generalized-ES norm in the next proposition.

\begin{proposition} 
 \label{pr-250406}
Let $g$ be a continuous distortion function, and denote $\alpha_j=g(j/n)$ for $j=0, \ldots, n$.
\begin{itemize}
 \item[{\rm (i)}] For $j=0, \ldots, n-1$, we have
   \begin{equation}
     \label{eq-250402-1}
      \langle\!\langle\bm x\rangle\!\rangle_{\alpha_j}^{S,g}=\frac{1}{1-\alpha_j} \sum_{i=j+1}^n c_i |x|_{(i)},
   \end{equation}
   where $c_i=g(i/n)-g((i-1)/n)=\alpha_i-\alpha_{i-1}$ for each $i$.

 \item[{\rm (ii)}] For $\alpha$ such that $\alpha_j<\alpha<\alpha_{j+1}$, $j=0, \ldots, n-2$, $\langle\!\langle\bm x\rangle\! \rangle_{\alpha}^{S, g}$ equals the weighted average of $\langle\! \langle \bm x \rangle\!\rangle_{\alpha_j}^{S, g}$ and $\langle\!\langle\bm x\rangle\!\rangle_{\alpha_{j+1}}^{S, g}$:
     \begin{equation}
       \label{eq-250402-2}
       \langle\!\langle\bm x\rangle\!\rangle_{\alpha}^{S,g}=\mu\langle\!\langle\bm x\rangle\!\rangle_{\alpha_j}^{S,g}
        +(1-\mu)\langle\!\langle\bm x\rangle\!\rangle_{\alpha_{j+1}}^{S, g},
     \end{equation}
    where $\mu=(\alpha_{j+1}-\alpha)(1-\alpha_j)/[(\alpha_{j+1}-\alpha_j)(1-\alpha)]\in(0,1)$, satisfying
    $$
          \frac{1}{1-\alpha}=\frac{\mu}{1-\alpha_j}+\frac{1-\mu}{1-\alpha_{j+1}}.
    $$

 \item[{\rm (iii)}] For $\alpha\in(\alpha_{n-1},1)$, we have
     \begin{equation*}
       \label{eq-250402-3}
        \langle\!\langle \bm x \rangle\!\rangle_\alpha^{S,g} = |x|_{(n)}.
     \end{equation*}
\end{itemize}
\end{proposition}

\begin{proof}
Without loss of generality, assume that all the $|x_i|$ are different. Let $X$ be a discrete random variable uniformly distributed on $\{x_1,\ldots,x_n\}$. Denote the left inverse of $g$ by
$$
   g^{-1}(t)=\inf\{x\in [0,1]: g(x)\ge t\},\quad t\in [0,1],
$$
and its right inverse by
$$
   g^{-1+}(t)=\inf\{x\in [0,1]: g(x)>t\},\ t\in [0,1),
$$
with $g^{-1+}(1)=1$. Similarly, the right-continuous inverse of $F_X$ is defined by $F_X^{-1+}(\alpha):=\inf\{x: F_X(x)>\alpha\}$ for $\alpha\in [0,1)$ with $F_X^{-1+}(1)={\rm esssup}(X)$. By Theorem 3.1 of \cite{GZGH24}, we have $u_{\alpha_j}=g^{-1}(\alpha_j)=g^{-1}(g(j/n))\le j/n$ and $v_{\alpha_j}=g^{-1+}(\alpha_j)\ge j/n$ since $g$ is continuous. Then
$$
   F_{|X|}^{-1}(u_{\alpha_j}) \le F_{|X|}^{-1}\(\frac {j}{n}\) =|x|_{(j)},\quad
    F_{|X|}^{-1+}(v_{\alpha_j}) \ge F_{|X|}^{-1+}\(\frac {j}{n}\) =|x|_{(j+1)}.
$$
Choose $t_{\alpha_j}:=|x|_{(j)}\in \[F_{|X|}^{-1}(u_{\alpha_j}),  F_{|X|}^{-1+}(v_{\alpha_j})\]$,
again by Theorem 3.1(3) of \cite{GZGH24}, we have
\begin{align}
  \langle\!\langle \bm x\rangle\!\rangle _{\alpha_j}^{S,g} & = \cR^g_{\alpha_j}(|X|)
   = t_{\alpha_j} + \frac{1}{1-\alpha_j}\sum_{i=1}^n c_i \(|x|_{(i)}-t_{\alpha_j}\)_+ \nonumber \\
  &= |x|_{(j)}+ \frac{1}{1-\alpha_j}\sum_{i=1}^n c_i \(|x|_{(i)}-|x|_{(j)}\)_+ \nonumber \\
  &= |x|_{(j)}+ \frac{1}{1-\alpha_j} \sum_{i=j+1}^n c_i \(|x|_{(i)}-|x|_{(j)}\) \nonumber \\
  &= |x|_{(j)}\(1-\frac{1}{1-\alpha_j}\sum_{i=j+1}^n c_i\)+\frac{1}{1-\alpha_j}\sum_{i=j+1}^n c_i|x|_{(i)}\nonumber \\
  & = \frac{1}{1-\alpha_j}\sum_{i=j+1}^n c_i|x|_{(i)}, \quad j=1,\ldots,n-1,     \label{eq-250402-4}
\end{align}
where the last equation follows from the fact that $\sum_{i=j+1}^n c_i=\sum_{i=j+1}^n [g(i/n)-g((i-1)/n)]  =g(1)- g(j/n)=1-\alpha_j$. When $\alpha=0$, \eqref{eq-250401-6} is identical to \eqref{eq-250402-1} with $j=0$.

For $\alpha$ such that $\alpha_{j}<\alpha<\alpha_{j+1}$ with $j\in \{0,\ldots,n-1\}$, since $g$ is increasing and continuous, it follows that $g^{-1}(p)$ and $g^{-1+}(p)$ are strictly increasing. Thus, we have $j/n < g^{-1}(\alpha)\le g^{-1+}(\alpha) < (j+1)/n$ and, hence,
$$
   F_{|X|}^{-1}(u_\alpha) \ge F_{|X|}^{-1}\(\frac {j}{n}+\)= |x|_{(j+1)},\quad
    F_{|X|}^{-1+}(v_{\alpha_j}) \le F_{|X|}^{-1}\(\frac {j+1}{n}-\) =|x|_{(j+1)}.
$$
Note that, in this case, the interval $\[F_{|X|}^{-1}(u_{\alpha_j}),  F_{|X|}^{-1+}(v_{\alpha_j})\]$ contains only one point $t_\alpha =|x|_{(j+1)}$. Again by Theorem 3.1(3) in \cite{GZGH24} and \eqref{eq-250402-4}, we have
\begin{equation*}
 \begin{aligned}
  \langle\!\langle \bm x \rangle\!\rangle_\alpha^{S,g} & = \cR^g_\alpha(|X|)
      = |x|_{(j+1)}+ \frac{1}{1-\alpha} \sum_{i=1}^n c_i \(|x|_{(i)}-|x|_{(j+1)}\)_+ \\
   & = |x|_{(j+1)} +\frac{\mu}{1-\alpha_j} \sum_{i=j+1}^n c_i \(|x|_{(i)} -|x|_{(j+1)}\)
          + \frac{1-\mu}{1-\alpha_{j+1}} \sum_{i=j+2}^n c_i \(|x|_{(i)} -|x|_{(j+1)}\)\\
   & = \mu \langle\!\langle \bm x \rangle\!\rangle_{\alpha_j}^{S,g} + (1-\mu)\langle\!\langle \bm x
         \rangle\! \rangle_{\alpha_{j+1}}^{S,g} + |x|_{(j+1)} \left(1-\frac{\mu}{1-\alpha_j}\sum_{i=j+1}^n c_i-\frac{1-\mu}{1-\alpha_{j+1}}\sum_{i=j+2}^{n} c_i\right)\\
   & = \mu \langle\!\langle \bm x\rangle\!\rangle_{\alpha_j}^{S,g}
            + (1-\mu)\langle\!\langle \bm x \rangle\!\rangle_{\alpha_{j+1}}^{S,g},
  \end{aligned}
\end{equation*}
where the last equality follows since
$$
  1-\frac{\mu}{1-\alpha_j}\sum_{i=j+1}^n c_i-\frac{1-\mu}{1-\alpha_{j+1}}\sum_{i=j+2}^{n} c_i
    =1- \frac{\mu}{1-\alpha_j}\cdot (1-\alpha_j) -\frac{1-\mu}{1-\alpha_{j+1}} (1-\alpha_{j+1})=0.
$$

For $\alpha\in (\alpha_{n-1},1)$, it is seen that the interval $\[F_{|X|}^{-1}(u_{\alpha}),  F_{|X|}^{-1+}(v_{\alpha})\]$ contains only one point $t_\alpha =|x|_{(n)}$. Then
$$
   \langle\!\langle \bm x \rangle\!\rangle_\alpha^{S,g} = |x|_{(n)}+\frac {1}{1-\alpha} \sum_{i=1}^n c_i \(|x|_{(i)}-|x|_{(n)}\)_+ =|x|_{(n)}.
$$
This completes the proof of the proposition.
\end{proof}

Finally, note that the scaled generalized-ES norm $\langle\!\langle \bm x \rangle\!\rangle_\alpha^{S,g}$ can be obtained by solving the following optimization problem:
\begin{equation}
 \label{eq-250506}
\begin{aligned}
    \langle\!\langle \bm x \rangle\!\rangle_\alpha^{S,g} &= \min_{t\in\R} \left\{t+\frac{1}{1-\alpha}\sum_{i=1}^n c_i z_i\right\} \\
        \text{s.t.}\quad & z_i\ge |x|_{(i)}-t, \ \  i=1,\ldots,n,\\
     & z_i\ge 0,\ \  i=1,\ldots,n.
\end{aligned}
\end{equation}

\begin{example} \normalfont
 \label{ex-250405}
  Let $\bm x=(-2, 1, 7, 10, -12)\in\R^5$ and choose a distortion function $g(u)=u^2$. Then $(|x|_{(1)}$, $|x|_{(2)}, \ldots, |x|_{(5)})= (1, 2, 7, 10, 12)$, $(\alpha_1, \ldots, \alpha_5)=(0.04, 0.16, 0.36, 0.64, 1)$ and $(c_1, c_2, c_3, c_4, c_5)= (0.04, 0.12, 0.2, 0.28, 0.36)$. Hence,
\begin{align*}
  \langle\!\langle \bm x\rangle\!\rangle_{0}^{S,g} &
        = 0.04\times 1 + 0.12\times 2 +0.2\times 7 +0.28\times 10 + 0.36\times 12=8.8,\\
  \langle\!\langle \bm x\rangle\!\rangle_{0.04}^{S,g} &
       = \frac {1}{1-0.04}\big [0.12\times 2 +0.2\times 7 +0.28\times 10 + 0.36\times 12\big ]=9.125,\\
  \langle\!\langle \bm x\rangle\!\rangle_{0.16}^{S,g} &
       = \frac {1}{1-0.16}\big [0.2\times 7 +0.28\times 10 + 0.36\times 12\big ]=10.143,\\
  \langle\!\langle \bm x\rangle\!\rangle_{0.36}^{S,g} &
       = \frac {1}{1-0.36}\big [0.28\times 10 + 0.36\times 12\big ]=11.125,\\
  \langle\!\langle \bm x\rangle\!\rangle_{\alpha}^{S,g} & =12,\quad \alpha\in [0.64, 1].
\end{align*}
   For $\alpha=0.5\in (\alpha_3, \alpha_4)$,  $\langle\!\langle \bm x\rangle\!\rangle_{\alpha}^{S,g}$ is the weighted average of $\langle\!\langle \bm x\rangle\!\rangle_{0.36}^{S,g}$ and $\langle\!\langle \bm x\rangle\!\rangle_{0.64}^{S,g}$ with respective weights
   $$
      \mu=\frac {(\alpha_4-\alpha)(1-\alpha_3)}{(\alpha_4-\alpha_3)(1-\alpha)} =0.64\ \ \hbox{and}\ \ 1-\mu=0.36.
   $$
   Thus, $\langle\!\langle \bm x\rangle\!\rangle_{0.5}^{S,g}= 0.64 \langle\!\langle \bm x\rangle\!\rangle_{0.36}^{S,g} + 0.36 \langle\!\langle \bm x\rangle\!\rangle_{0.64}^{S,g}=0.64\times 11.125+0.36\times 12=11.44$.
\end{example}

\begin{example} \normalfont
 \label{ex-250406}
    Consider the scaled generalized-ES norms $\langle\!\langle \cdot\rangle\!\rangle_\alpha^{S,g_1}$, $\langle\!\langle \cdot\rangle\!\rangle_\alpha^{S,g_2}$ and $\langle\!\langle \cdot \rangle\!\rangle_\alpha^{S,g_3}$ in $\R^2$, where distortion functions $g_1(u)=u^2$, $g_2(u)=u$ and $g_3(u)=u^{1/2}$ for $u\in (0,1)$. In view of the monotonicity of the generalized-ES norms $\langle\!\langle \bm x\rangle\!\rangle_\alpha^{S,g_i}$  with respect to $\alpha$ for fixed $\bm x\in\R^2$, it can be checked that
\begin{align*}
  \langle\!\langle \bm x\rangle\!\rangle_\alpha^{S,g_1} &=\left \{\begin{array}{ll}
    \dfrac {1}{4} |x|_{(1)} +\dfrac {3}{4} |x|_{(2)}, & \alpha=0,\\[8pt]
    \dfrac {1-4\alpha}{4(1-\alpha)} |x|_{(1)}+\dfrac {3}{4(1-\alpha)} |x|_{(2)}, & \alpha\in \(0, \dfrac {1}{4}\),\\[8pt]
    |x|_{(2)}, & \alpha\in \[\dfrac {1}{4}, 1\];          \end{array} \right.
\end{align*}
\begin{align*}
  \langle\!\langle \bm x\rangle\!\rangle_\alpha^{S,g_2} &=\left \{\begin{array}{ll}
    \dfrac {1}{2} (|x_1| +|x_2|), & \alpha=0,\\[8pt]
    \dfrac {1-2\alpha}{2(1-\alpha)} |x|_{(1)}+\dfrac {1}{2(1-\alpha)} |x|_{(2)}, & \alpha\in \(0, \dfrac {1}{2}\),\\[8pt]
    |x|_{(2)}, & \alpha\in \[\dfrac {1}{2}, 1\];          \end{array} \right.
\end{align*}
and
\begin{align*}
  \langle\!\langle \bm x\rangle\!\rangle_\alpha^{S,g_3} &=\left \{\begin{array}{ll}
    \dfrac {\sqrt{2}}{2} |x|_{(1)} +\dfrac {2-\sqrt{2}}{2} |x|_{(2)}, & \alpha=0,\\[8pt]
    \dfrac {\sqrt{2}-2\alpha}{2(1-\alpha)} |x|_{(1)}+\dfrac {2-\sqrt{2}}{2(1-\alpha)} |x|_{(2)}, & \alpha\in \(0, \dfrac {\sqrt{2}}{2}\),\\[8pt]
    |x|_{(2)}, & \alpha\in \[\dfrac {\sqrt{2}}{2}, 1\].          \end{array} \right.
\end{align*}

\begin{figure}[htbp]
  \centering
 \subfigure[\ ]{ \includegraphics[scale=0.35]{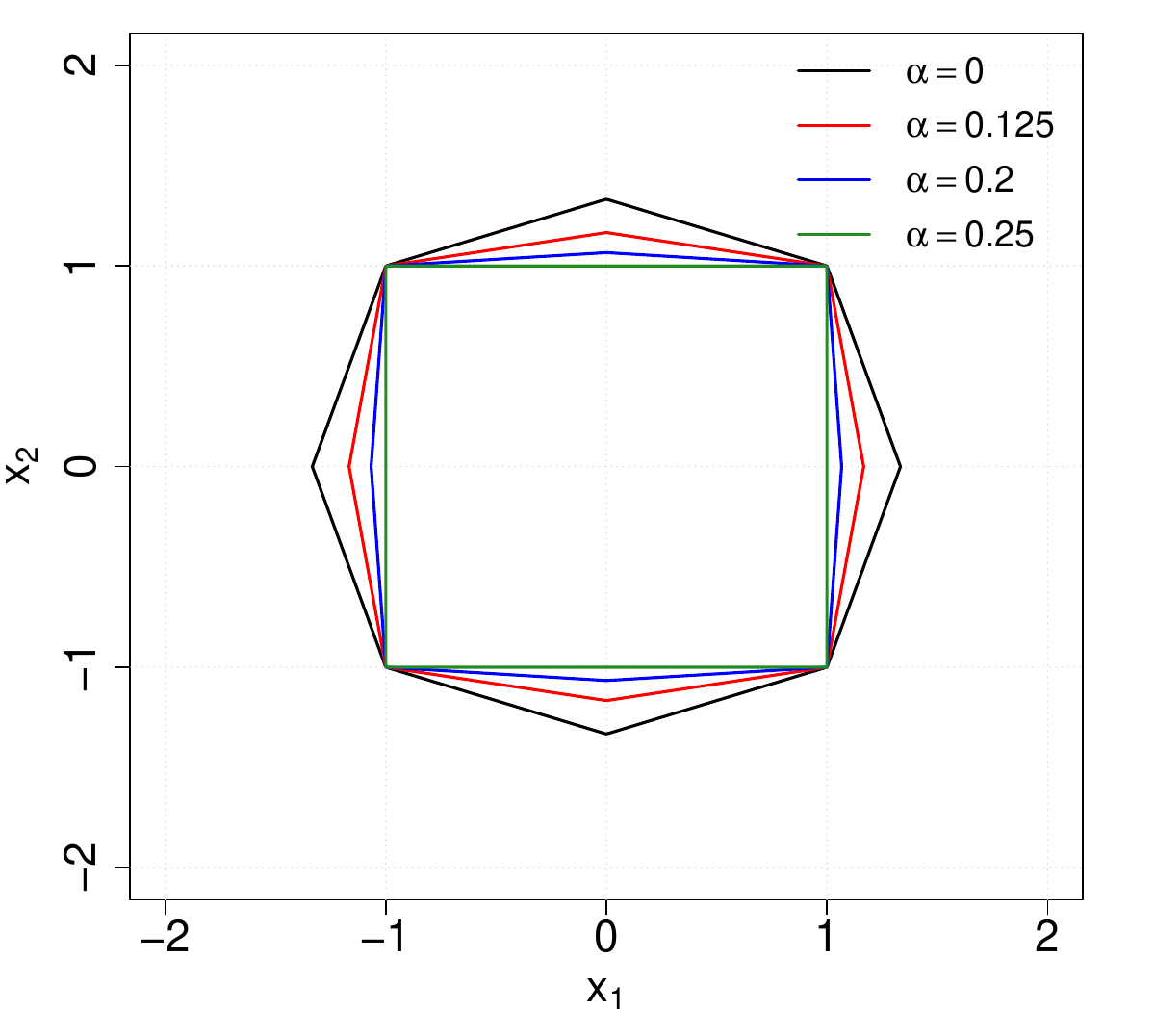} }
  \subfigure[\ ]{ \includegraphics[scale=0.35]{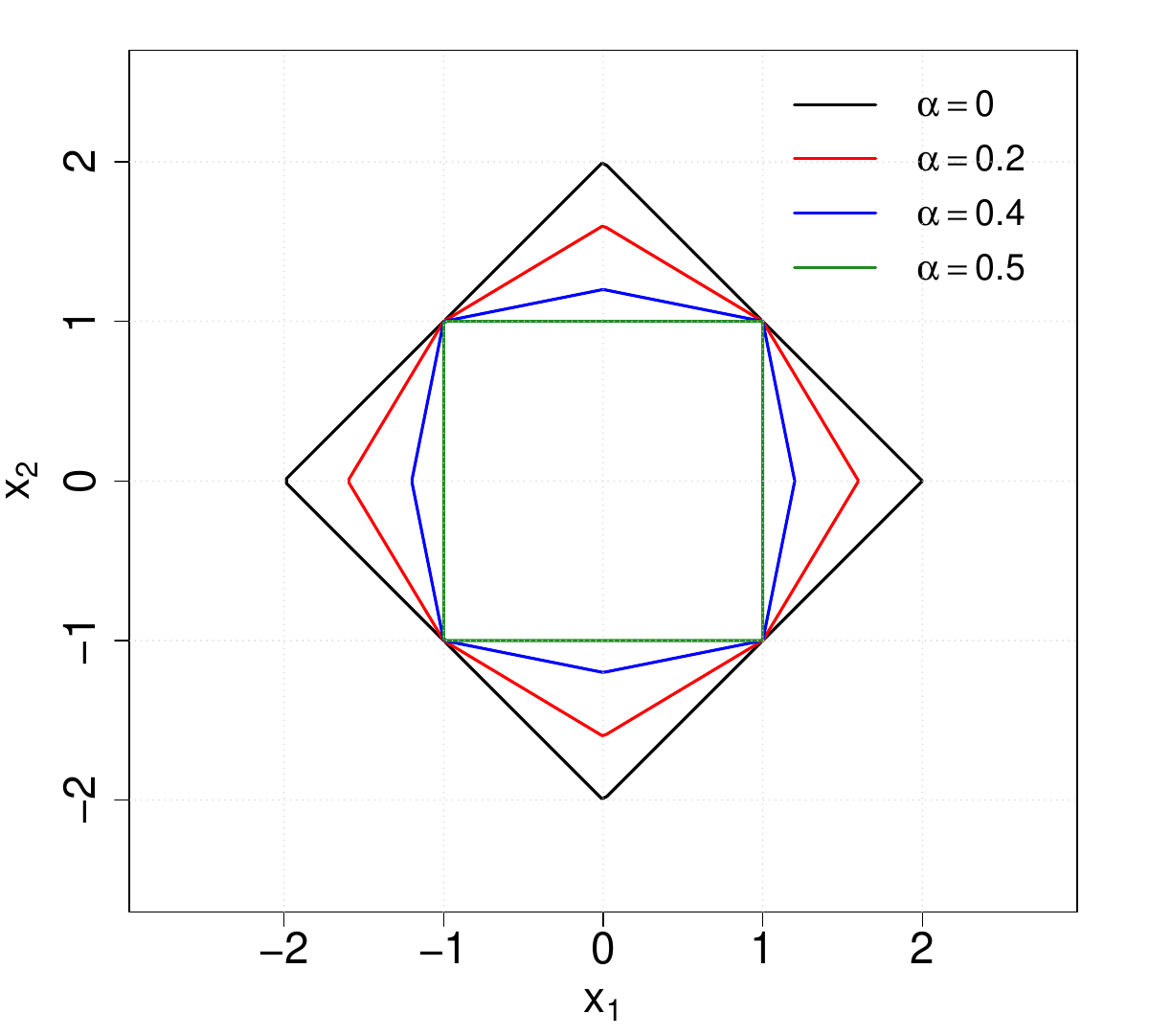} }
  \subfigure[\ ]{ \includegraphics[scale=0.35]{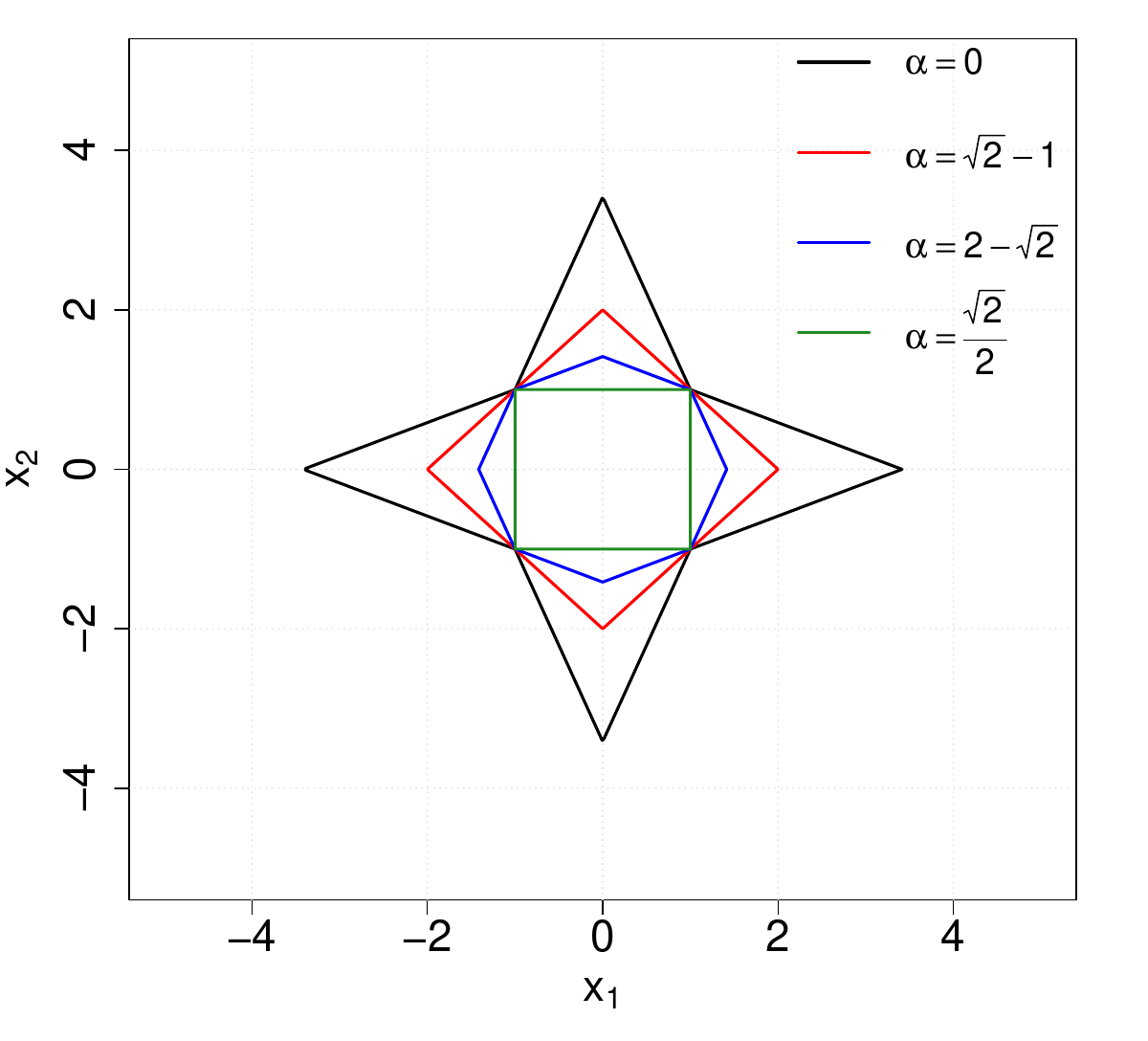} }
  \subfigure[\ ]{ \includegraphics[scale=0.35]{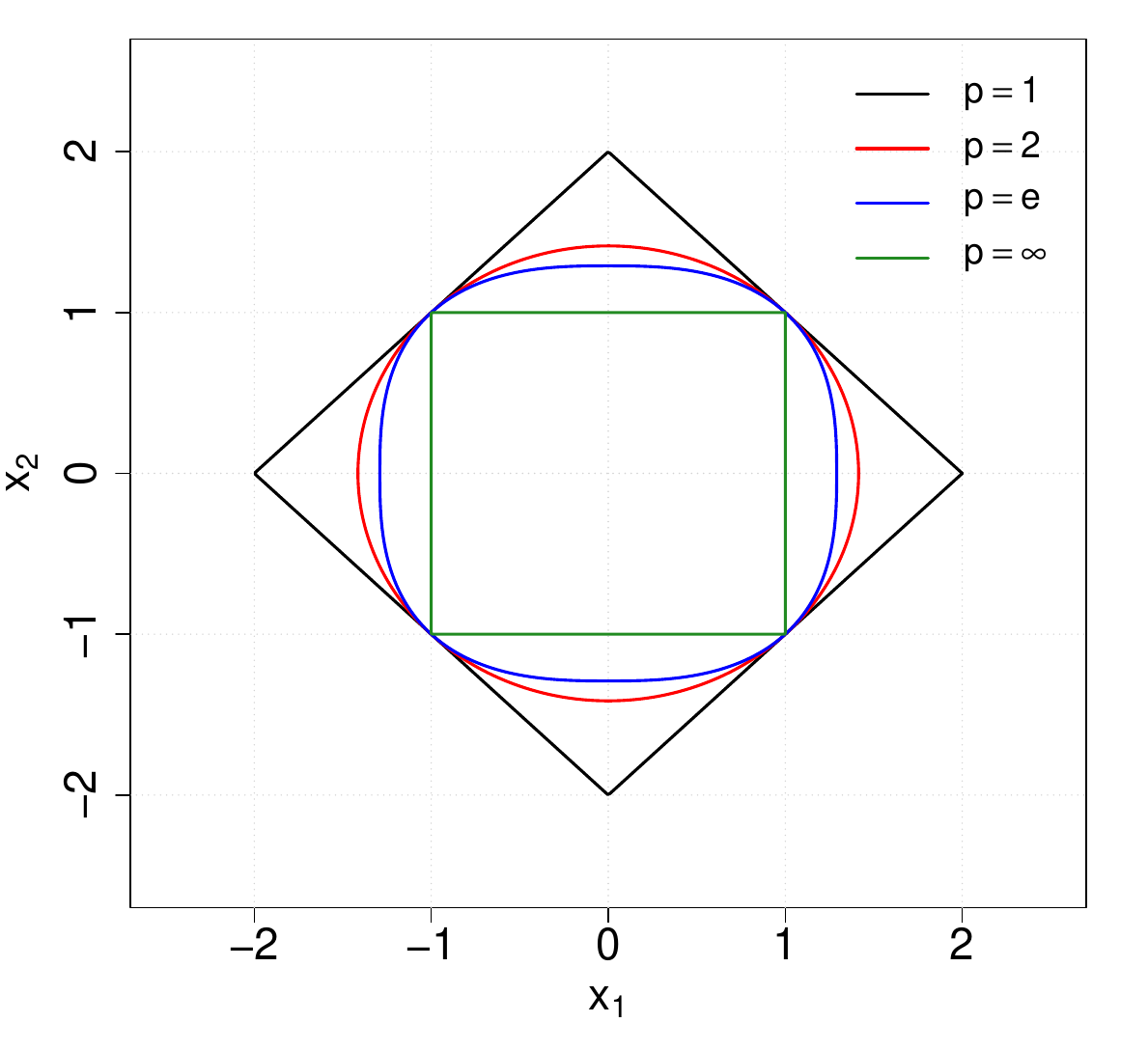} }
  \caption{(a) Unit disk $\mathcal{U}_{S,g_1}$ with $g_1(u) = u^2$ and $\alpha \in \{0, 0.125, 0.2, 0.25\}$; (b) Unit disk $\mathcal{U}_{S,g_2}$ with $g_2(u) = u$ and $\alpha \in \{0, 0.2, 0.4, 0.5\}$; (c) Unit disk $\mathcal{U}_{S,g_3}$ with $g_3(u) = \sqrt{u}$ and $\alpha \in\{0, \sqrt{2}-1, 2-\sqrt{2}, \sqrt{2}/2\}$; (d) Unit disk $\mathcal{U}_{\|\cdot\|_p^S}$ based on the scaled $L^S_p$-norm with $p\in \{1, 2, e, \oo\}$. }
   \label{Fig-SGES-disk}
\end{figure}

   A unit disk (ball) for a norm $\|\cdot\|$ in $\R^n$ is a set of vectors with their norm less than or equal $1$, i.e., $\mathcal{U}_{\|\cdot\|}: =\{\bm x\in\R^n: \|\bm x\| \le 1\}$. Denote by $\mathcal{U}_{S,g}$ the unit disk based on the scaled generalized-ES norm $\langle\!\langle \cdot\rangle\!\rangle_\alpha^{S,g}$. Figure \ref{Fig-SGES-disk} depicts the unit disks $\mathcal{U}_{S,g_1}$, $\mathcal{U}_{S,g_2}$ and $\mathcal{U}_{S,g_3}$ with different $\alpha$ and the unit disk  $\mathcal{U}_{\|\cdot\|_p^S}$ based on the scaled $L_p^S$-norm $\|\cdot\|_p^S$ with different $p\ge 1$, where
   $$
     \|\bm x\|_p^S =\(\frac {1}{n}\sum^n_{i=1} |x_i|^p\)^{1/p},\quad \bm x\in\R^n.
   $$
   The figure shows that when the distortion function is $g_1(u)=u^2$, the unit disk $\mathcal{U}_{S,g_1}$ deforms continuously from a regular octagon to the unit disk of $L_\infty$-norm (a regular quadrilateral) as $\alpha$ gradually increases from $0$ to $0.25$. For the case of $g_3(u)=\sqrt{u}$, as $\alpha$ increases from $0$, the unit disk $\mathcal{U}_{S,g_3}$ initially manifests as a star-shaped octagon and remains non-convex until $\alpha$ reaches $\sqrt{2}-1$, at which point the unit disk becomes a square diamond (i.e. a regular rhombus). Then as $\alpha$ further increases, the shape evolves into a convex octagon and gradually transforms into a regular octagon when $\alpha=2-\sqrt{2}$. Eventually, it also deforms into the regular quadrilateral as $\alpha$ continues to grow to $\sqrt{2}/2$. Such a somewhat peculiar transformation behavior arises from the concavity of the distortion function $g_3$.
\end{example}

From Figure \ref{Fig-SGES-disk}, it can be seen that $\mathcal{U}_{\|\cdot\|_p^S}$, $\mathcal{U}_{S,g_1}$ and  $\mathcal{U}_{S,g_2}$ are convex sets, while $\mathcal{U}_{S,g_3}$ is not a convex set when $\alpha=0$. The next proposition gives a condition upon distortion function $g$ under which $\mathcal{U}_{S,g}$ is convex.

\begin{proposition}
  \label{pr-250405}
If $g(u)$ is convex on $[0,1]$, then the unit disk $\mathcal{U}_{S,g}$ based on scaled generalized-ES norm $\langle\!\langle \cdot\rangle\!\rangle_\alpha^{S,g}$ is a convex set in $\R^n$. If $g(u)$ is concave on $[0,1]$, then $\mathcal{U}_{S,g}$ is in general not a convex set.
\end{proposition}

\begin{proof}
First, assume $g(u)$ is convex in $u$. Then, by Proposition 2.1 in \cite{GZGH24}, $\cR_\alpha$ possesses the convexity for $\alpha\in [0,1)$, that is, $\cR^g_\alpha(\beta X + (1-\beta) Y) \le \beta\cR^g_\alpha(X) + (1-\beta)\cR^g_\alpha(Y)$ for $X, Y\in L^\oo$ and $\beta\in [0,1]$. Choose any $\bm x, \bm y\in \mathcal{U}_{S,g}$,  that is, $ \langle\!\langle \bm x \rangle\!\rangle_\alpha^{S,g}\le 1$ and $ \langle\!\langle \bm y \rangle\!\rangle_\alpha^{S,g}\le 1$. Define two random variables $X=x_I$ and $Y=y_I$, where $I$ follows the uniform distributions on $\{1, \ldots, n\}$. Then, for any $\beta\in [0,1]$,
\begin{align*}
  \langle\!\langle \beta \bm x + (1-\beta)\bm y \rangle\!\rangle_\alpha^{S,g}
    & = \cR^g_{\alpha} (|\beta X+(1-\beta) Y|) \le \cR^g_{\alpha} (\beta |X| +(1-\beta) |Y|) \\
    & \le \beta \cR^g_{\alpha} (|X|) +  (1-\beta) \cR^g_{\alpha} (|Y|)\\
     & \le \beta \langle\!\langle \bm x \rangle\!\rangle_\alpha^{S,g}
        + (1-\beta) \langle\!\langle \bm y \rangle\!\rangle_\alpha^{S,g}\le 1,
\end{align*}
implying $\beta\bm x+ (1-\beta)\bm y\in \mathcal{U}_{S,g}$. This proves that $\mathcal{U}_{S,g}$ is a convex set.

Next, assume $g(u)$ is concave in $u$. A counterexample can be seen from Figure \ref{Fig-SGES-disk}(c) with $g(u)=u^{1/2}$ and $\alpha=0$.
\end{proof}

\section{Non-scaled generalized-ES norm}
\label{sect-3}

\begin{definition}
  The non-scaled generalized-ES norm $\langle\!\langle\cdot \rangle\!\rangle_\alpha^g$ in $\R^n$ with parameter $\alpha\in [0,1)$ and distortion function $g$ is defined by
  $$
     \langle\!\langle\bm x \rangle\!\rangle_\alpha^g =n (1-\alpha) \langle\!\langle\bm x \rangle\!\rangle_\alpha^{S,g}, \quad \bm x\in\R^n.
  $$
\end{definition}

Similar to Proposition \ref{pr-250406} for the scaled generalized-ES norm, we have the next proposition concerning alternative representation of the non-scaled generalized-ES norm. The proof is straightforward and, hence, omitted.

\begin{proposition} 
 \label{pr-250407}
Let $g$ be a continuous distortion function, and denote $\alpha_j=g(j/n)$ for $j=0, \ldots, n$. Then
\begin{itemize}
 \item[{\rm (i)}] For $j=0, \ldots, n-1$,
   \begin{equation}
     \label{eq-250407-1}
      \langle\!\langle\bm x\rangle\!\rangle_{\alpha_j}^g= n \sum_{i=j+1}^n c_i |x|_{(i)},\quad \bm x\in\R^n,
   \end{equation}
   where $c_i=g(i/n)-g((i-1)/n)=\alpha_i-\alpha_{i-1}$ for each $i$.

 \item[{\rm (ii)}] For $\alpha$ such that $\alpha_j<\alpha<\alpha_{j+1}$, $j=0, \ldots, n-2$, $\langle\!\langle\bm x\rangle\! \rangle_{\alpha}^g$ equals the weighted average of $\langle\! \langle \bm x \rangle\!\rangle_{\alpha_j}^g$ and $\langle\!\langle\bm x\rangle\!\rangle_{\alpha_{j+1}}^g$:
     \begin{equation}
       \label{eq-250407-2}
       \langle\!\langle\bm x\rangle\!\rangle_{\alpha}^g=\lm\langle\!\langle\bm x\rangle\!\rangle_{\alpha_j}^g
        +(1-\lm)\langle\!\langle\bm x\rangle\!\rangle_{\alpha_{j+1}}^g,\quad \bm x\in\R^n,
     \end{equation}
    where $\lm=(\alpha_{j+1}-\alpha) / (\alpha_{j+1}-\alpha_j)$.

 \item[{\rm (iii)}] For $\alpha\in((n-1)/n,1)$, we have $\langle\!\langle \bm x \rangle\!\rangle_\alpha^g = n(1-\alpha) |x|_{(n)}$ for $\bm x\in\R^n$.
\end{itemize}
\end{proposition}

\begin{example} \normalfont
  \label{ex-250407}
  Consider the generalized-ES norms $\langle\!\langle \cdot\rangle\!\rangle_\alpha^{g_1}$, $\langle\!\langle \cdot\rangle\!\rangle_\alpha^{g_2}$ and $\langle\!\langle \cdot \rangle\!\rangle_\alpha^{g_3}$ in $\R^2$, where distortion functions $g_1(u)=u^2$, $g_2(u)=u$ and $g_3(u)=u^{1/2}$ for $u\in (0,1)$. From Example \ref{ex-250406}, we have
  \begin{figure}[htbp]
    \centering
     \subfigure[\ ]{ \includegraphics[scale=0.35]{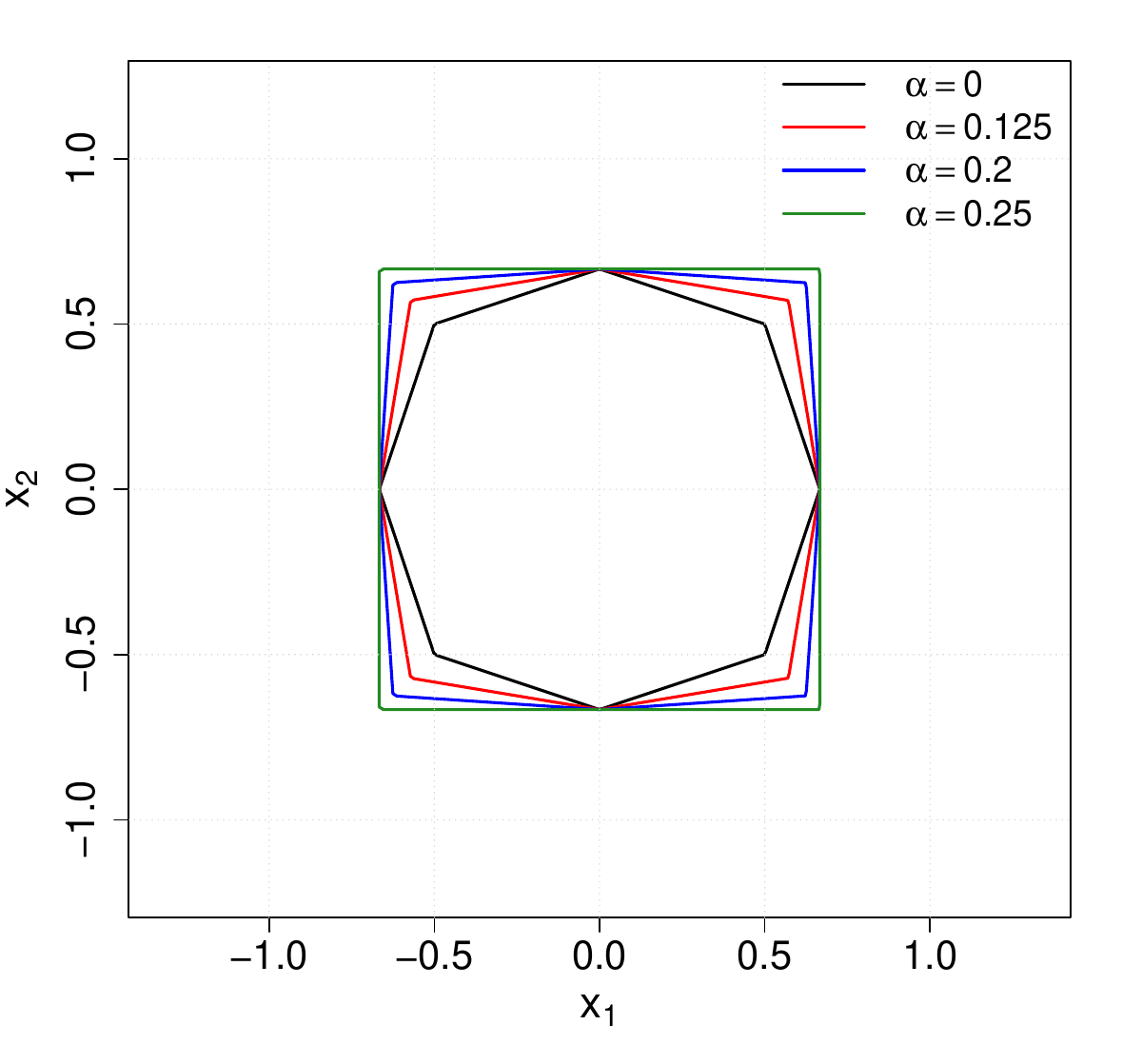} }
     \subfigure[\ ]{ \includegraphics[scale=0.35]{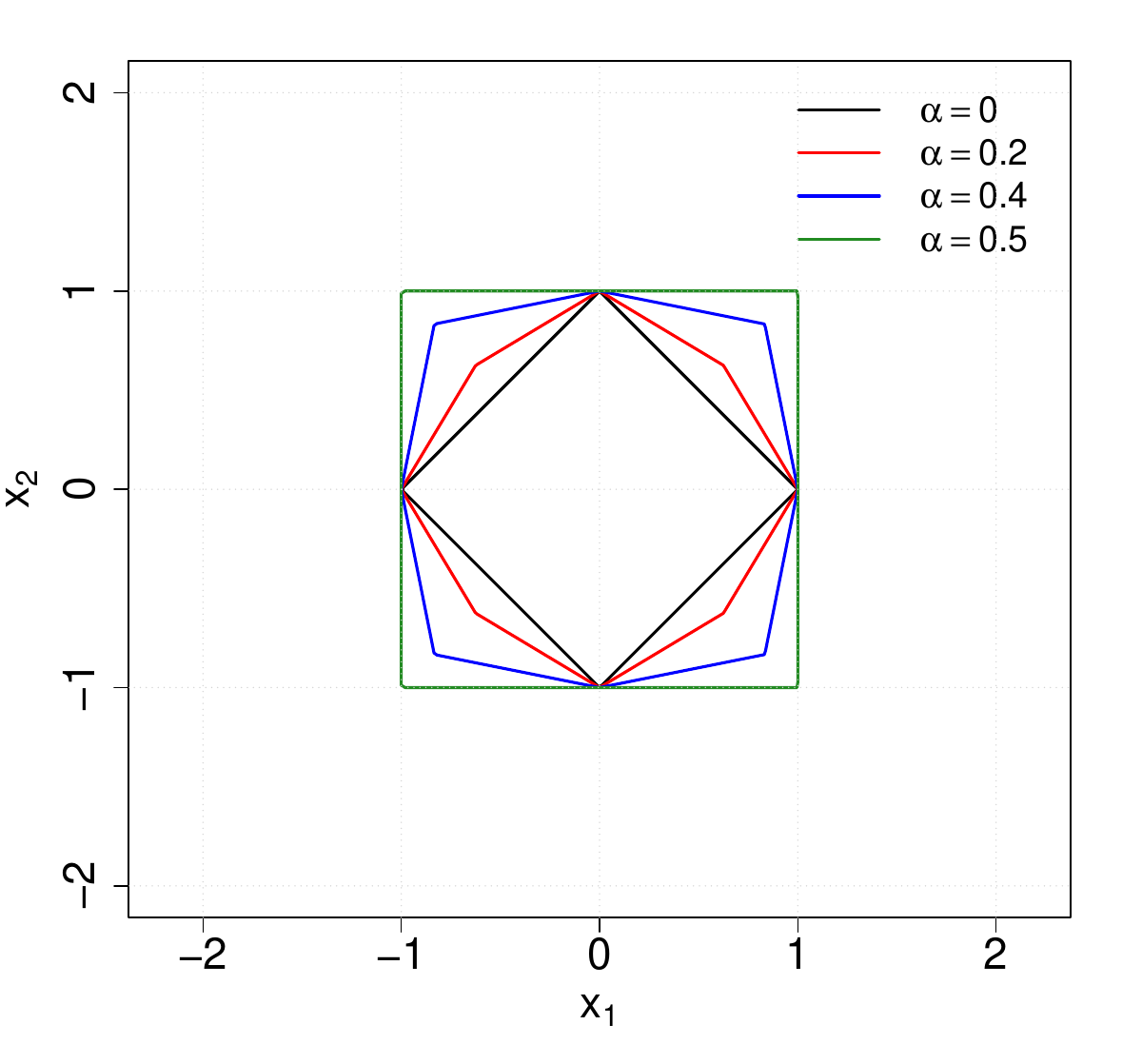} }
     \subfigure[\ ]{ \includegraphics[scale=0.35]{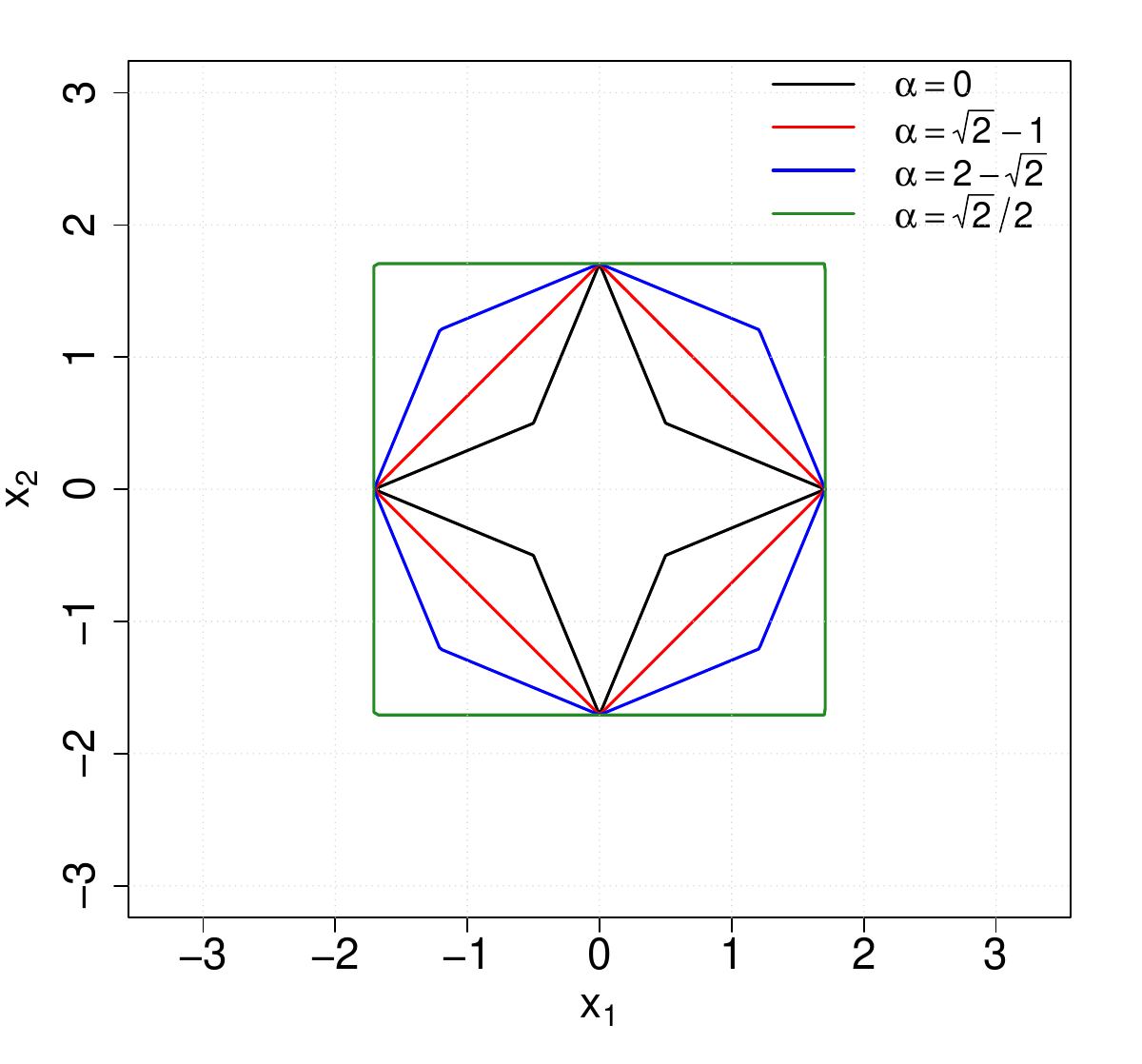} }
     \subfigure[\ ]{ \includegraphics[scale=0.35]{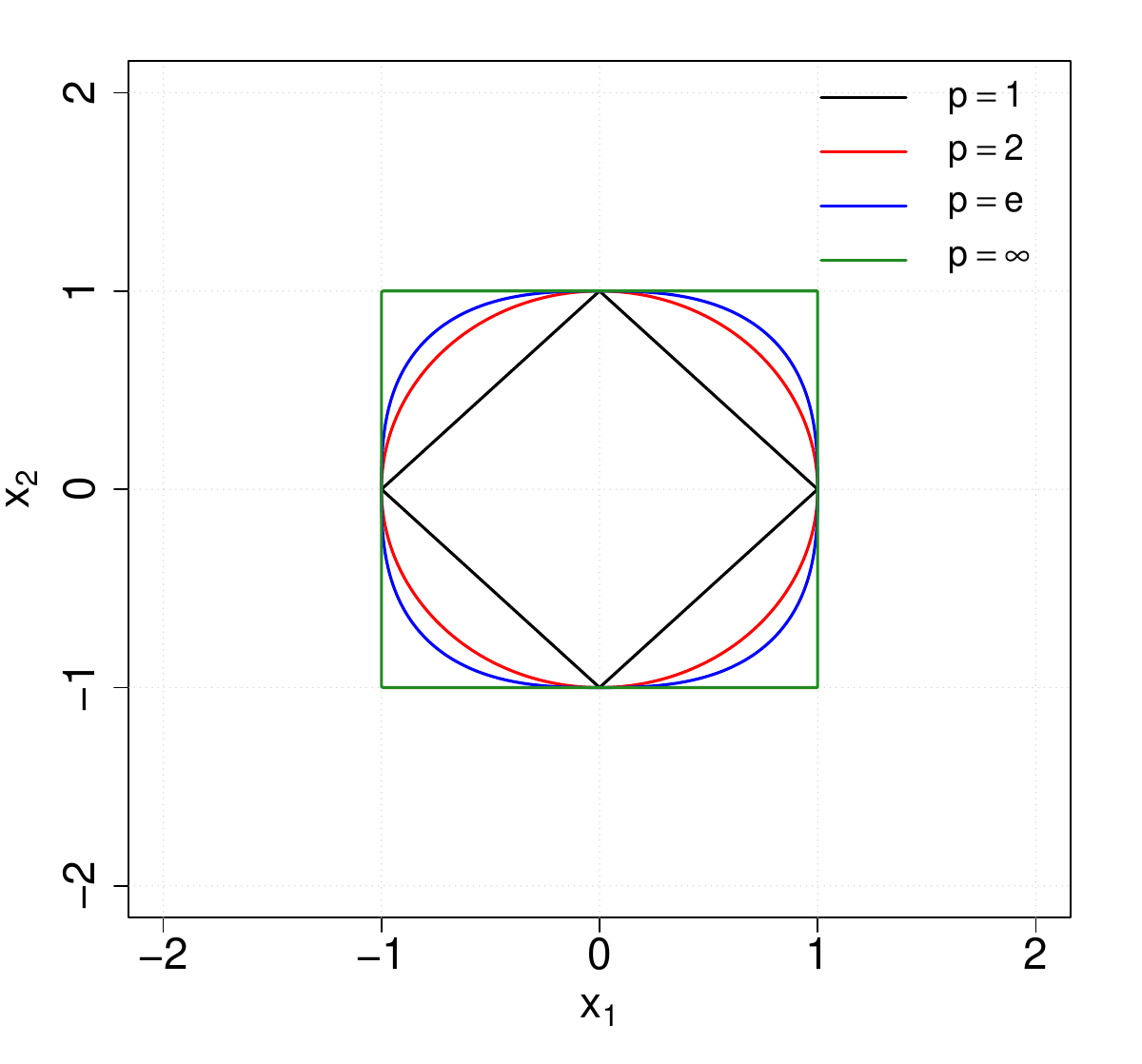} }
     \caption{(a) Unit disk $\mathcal{U}_{g_1}$ with $g_1(u) = u^2$ and $\alpha \in \{0, 0.125, 0.2, 0.25\}$; (b) Unit disk $\mathcal{U}_{g_2}$ with $g_2(u) = u$ and $\alpha \in \{0, 0.2, 0.4, 0.5\}$; (c) Unit disk $\mathcal{U}_{g_3}$ with $g_3(u) = \sqrt{u}$ and $\alpha \in\{0, \sqrt{2}-1, 2-\sqrt{2}, \sqrt{2}/2\}$; (d) Unit disk $\mathcal{U}_{\|\cdot\|_p}$ based on $L_p$-norm with $p\in \{1, 2, e, \oo\}$. }
   \label{Fig-NSGES-disk}
  \end{figure}
  \begin{align*}
    \langle\!\langle \bm x\rangle\!\rangle_\alpha^{g_1} &=\left \{\begin{array}{ll}
      \dfrac {1}{2} |x|_{(1)} +\dfrac {3}{2} |x|_{(2)}, & \alpha=0,\\[8pt]
     \dfrac {1}{2} (1-4\alpha) |x|_{(1)}+\dfrac {3}{2} |x|_{(2)}, & \alpha\in \(0, \dfrac {1}{4}\),\\[8pt]
          2(1-\alpha) |x|_{(2)}, & \alpha\in \[\dfrac {1}{4}, 1\];          \end{array} \right.\\[2pt]
   \langle\!\langle \bm x\rangle\!\rangle_\alpha^{S,g_2} &=\left \{\begin{array}{ll}
        |x_1| + |x_2|, & \alpha=0,\\[8pt]
        (1-2\alpha) |x|_{(1)}+ |x|_{(2)}, & \alpha\in \(0, \dfrac {1}{2}\),\\[8pt]
       2(1-\alpha) |x|_{(2)}, & \alpha\in \[\dfrac {1}{2}, 1\];          \end{array} \right.
  \end{align*}
  and
  \begin{align*}
     \langle\!\langle \bm x\rangle\!\rangle_\alpha^{g_3} &=\left \{\begin{array}{ll}
        \sqrt{2} |x|_{(1)} + (2-\sqrt{2}) |x|_{(2)}, & \alpha=0,\\[4pt]
     (\sqrt{2}-2\alpha) |x|_{(1)}+ (2-\sqrt{2}) |x|_{(2)}, & \alpha\in \(0, \dfrac {\sqrt{2}}{2}\),\\[10pt]
          2(1-\alpha) |x|_{(2)}, & \alpha\in \[\dfrac {\sqrt{2}}{2}, 1\].     \end{array} \right.
  \end{align*}

  Figure \ref{Fig-NSGES-disk} compares the unit disks for different non-scaled generalized-ES norms with different distortions functions. For $\bm x=(3,1,-4, 18,10)$, Figure \ref{Fig-GES-comparison} compares the scaled and non-scaled norms $\langle\!\langle\bm x\rangle\!\rangle_{\alpha}^{S,g_i}$ and $\langle\!\langle\bm x\rangle\!\rangle_{\alpha}^{g_i}$ numerically.
  \begin{figure}[htbp]
    \centering
    \subfigure[\ ]{ \includegraphics[scale=0.35]{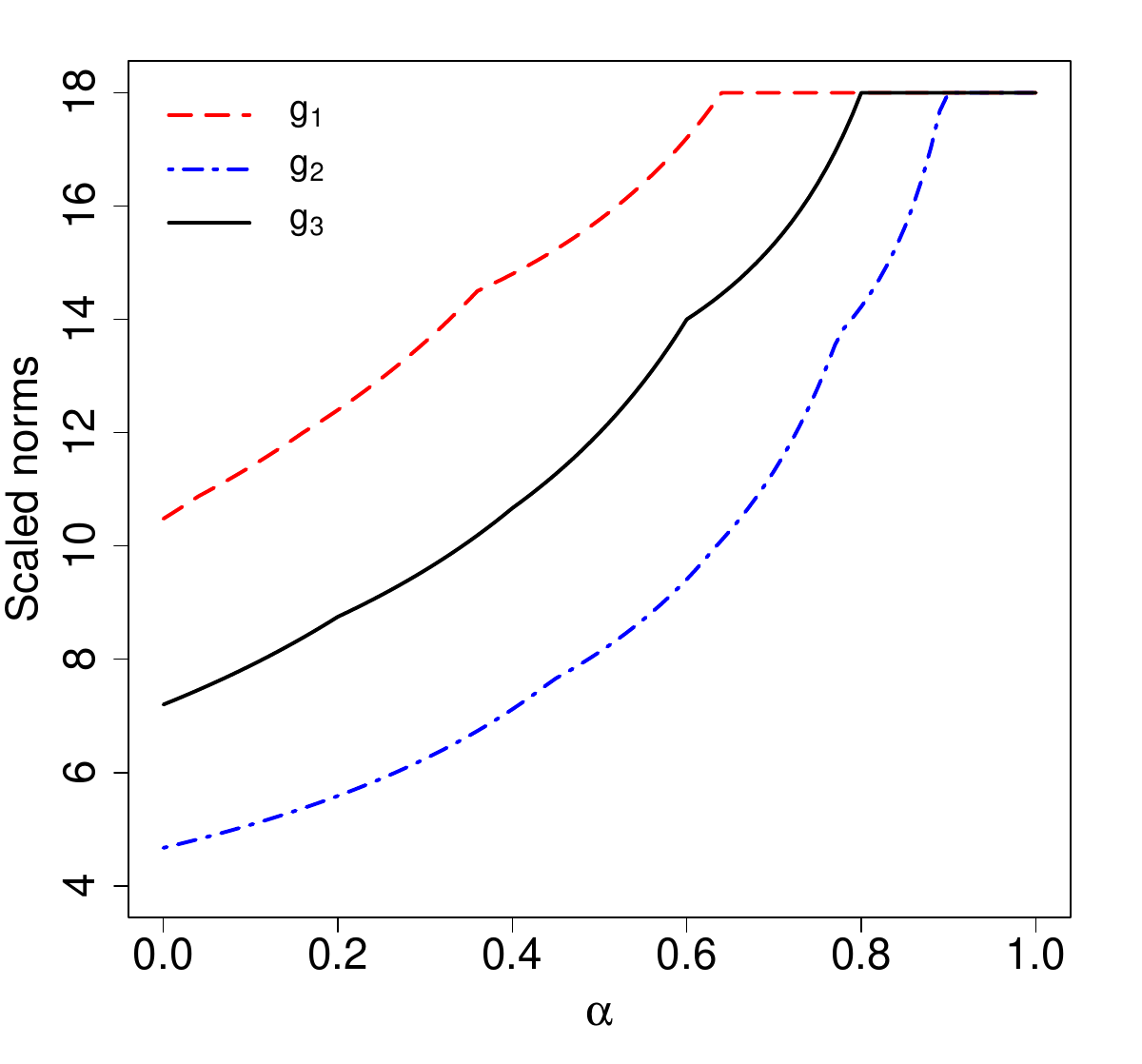} }
    \subfigure[\ ]{ \includegraphics[scale=0.35]{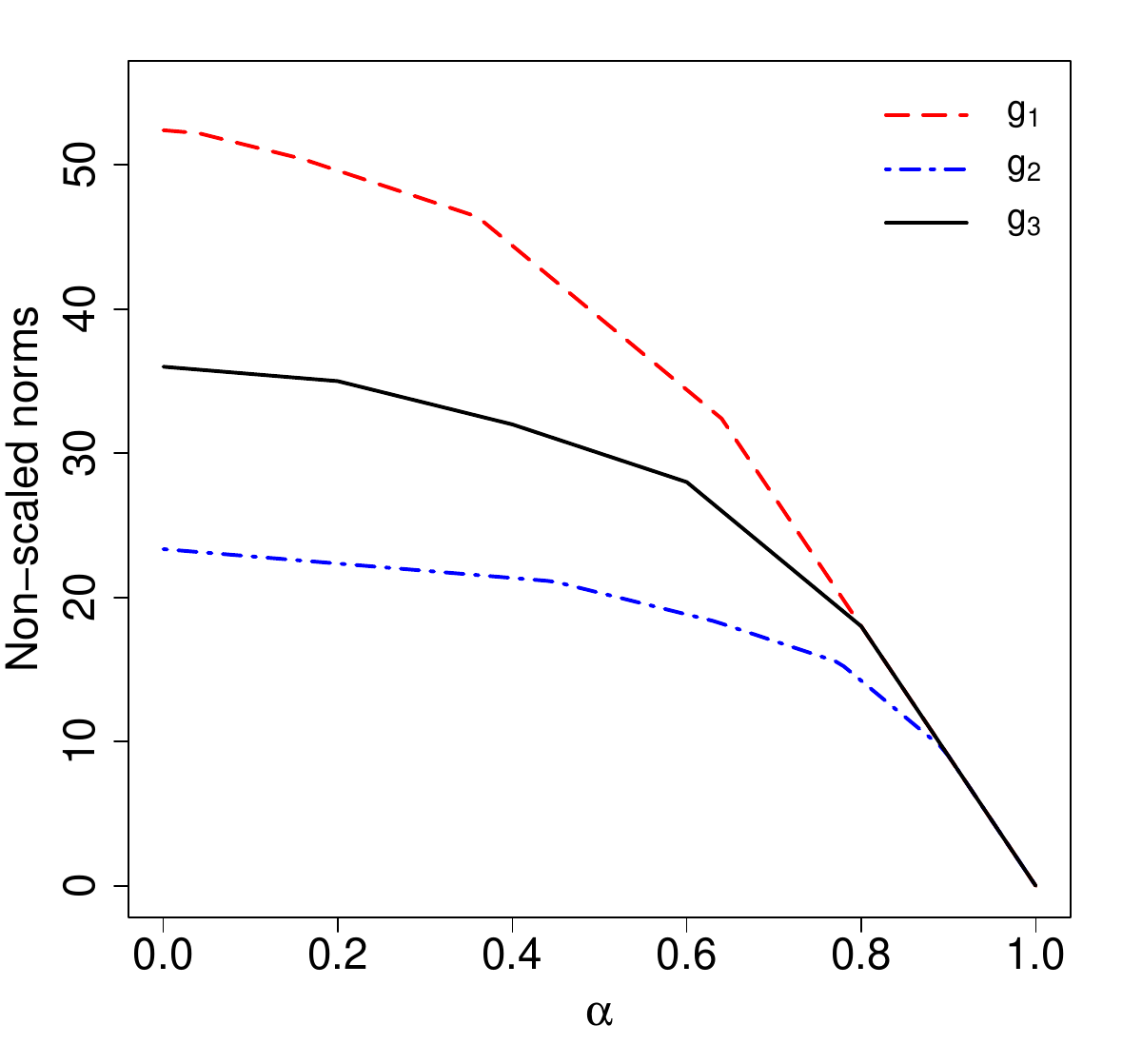} }
     \caption{(a) The graphs of scaled norms $\langle\!\langle\bm x\rangle\!\rangle_{\alpha}^{S,g_i}$ for $i=1,2,3$;  (b) The graphs of non-scaled norms $\langle\!\langle\bm x\rangle\!\rangle_{\alpha}^{g_i}$ for $i=1,2,3$. Here $\bm x=(3,1,-4, 18,10)$. }
    \label{Fig-GES-comparison}
  \end{figure}
\end{example}

\begin{proposition} 
  \label{pr-250408}
  Let $g$ be a continuous distortion function. For $\bm x\in\R^n$, the non-scaled generalized-ES norm $\langle\!\langle \bm x \rangle\!\rangle_\alpha^g$ is decreasing and piecewise-linear in $\alpha$. Moveover, $\langle\!\langle \bm x \rangle\!\rangle_\alpha^g$ is concave in $\alpha$ if $g$ is convex.
\end{proposition}

\begin{proof}
Let $\alpha_j$ and $c_j$ be as defined in Proposition \ref{pr-250407}. For $j=0,1,\ldots, n-1$, we have
$$
  \langle\!\langle \bm x \rangle\!\rangle_{\alpha_j}^g = n \sum_{i=j+1}^n c_i |x|_{(i)}
      \ge n \sum_{i=j+2}^n c_i |x|_{(i)} =\langle\!\langle \bm x \rangle\!\rangle_{\alpha_{j+1}}^g,\quad \bm x\in\R^n.
$$
By Proposition \ref{pr-250407} (ii), for $\alpha_j<\alpha<\alpha_{j+1}$, $\langle\!\langle\bm x\rangle\! \rangle_{\alpha}^g$ is a convex combination of $\langle\! \langle \bm x \rangle\!\rangle_{\alpha_j}^g$ and $\langle\!\langle\bm x\rangle\!\rangle_{\alpha_{j+1}}^g$ with $\lm=(\alpha_{j+1}-\alpha) / (\alpha_{j+1}-\alpha_j)$, linear with respect to $\alpha$. Thus, $\langle\!\langle \bm x \rangle\!\rangle_\alpha^g$ is decreasing and piecewise-linear in $\alpha$.

To prove the concavity of $\langle\!\langle \bm x \rangle\!\rangle_\alpha^g$ in $\alpha$ when $g$ is convex, it suffices to show that
$$
   \langle\!\langle \bm x \rangle\!\rangle_{\alpha_j}^g\ge \frac {1}{2} \(\langle\!\langle \bm x \rangle\!\rangle_{\alpha_{j-1}}^g + \langle\!\langle \bm x \rangle\!\rangle_{\alpha_{j+1}}^g\),\quad j=1, \ldots, n-1.
$$
In fact, we have
\begin{align*}
  \langle\!\langle \bm x \rangle\!\rangle_{\alpha_{j-1}}^g + \langle\!\langle \bm x \rangle\!\rangle_{\alpha_{j+1}}^g
   &=n \sum_{i=j}^n c_i |x|_{(i)} + n \sum_{i=j+2}^n c_i |x|_{(i)}\\
   & = n\[c_j |x|_{(j)} - c_{j+1} |x|_{(j+1)}\] + 2 n \sum_{i=j+1}^n c_i |x|_{(i)}\\
   & \le 2 n \sum_{i=j+1}^n c_i |x|_{(i)} =2  \langle\!\langle \bm x \rangle\!\rangle_{\alpha_j}^g,
\end{align*}
where the inequality follows since the convexity of $g$ implies $c_j \le c_{j+1}$. This prove the proposition.
\end{proof}

\section{Dual norm of the scaled generalized-ES norm}
\label{sect-4}

The dual norm is a fundamental concept in functional analysis that characterizes the structure of normed spaces. Given a norm $\|\cdot\|$ on $\R^n$, its dual norm $\|\cdot\|^\ast$ is defined as
\begin{equation*}
    \|\bm y\|^\ast =\sup_{\|\bm x\|\leqslant 1} \bm x^\top\bm y, \qquad \bm y\in \R^n.
\end{equation*}

For instance, when $p,\, q>1$ such that $1/p+1/q=1$, the dual norm of $L_p$-norm is the $L_q$-norm. To investigate the dual norm of scaled generalized-ES norm, we recall the concept of ordered weighted $L_1$ (OWL) norm introduced by \cite{ZF15}. Denoted by $\Omega_\mathbf{w}: \R^n\rightarrow \R_+$ the OWL norm, defined by
\begin{equation}
    \Omega_{\bm w}(\bm x) =\sum_{i=1}^n w_i|x|_{(i)},
\end{equation}
where $\bm w =(w_1, \ldots, w_n)\in\R^n$ satisfies $0\leqslant w_1\leqslant\ldots\leqslant w_n$. The next lemma gives the dual norm of the OWL norm.

\begin{lemma} {\rm \citep[Lemma 1]{ZF15}}
\label{dual-lemma}
 The dual norm of $\Omega_{\bm w}(x)$ is given by
 \begin{equation*}
     \Omega_{\bm w}^\ast (\bm y) = \max\left\{\tau_i \sum_{k=i}^n |y|_{(k)},\ \ i=1, \ldots, n\right\},\qquad \bm y\in\R^n,
 \end{equation*}
 where $\tau_i =\(\sum_{k=i}^n w_i\)^{-1}$ for $i=1, \ldots, n$.
\end{lemma}

As observed, the scaled generalized-ES norm $\langle\!\langle\cdot\rangle\!\rangle_\alpha^{S,g}$,  constructed using a convex distortion function $g$, can actually be interpreted as a special case of the OWL norm. Let $\{\alpha_j\}$ and $\{c_j\}$ be as defined in Proposition \ref{pr-250407}. Specifically, for $\alpha\in [\alpha_j,\alpha_{j+1}],$ with $ j=0,\ldots, n-2$, the scaled generalized-ES norm $\langle\!\langle\bm x\rangle\!\rangle_\alpha^{S,g}$ can be rewritten as
\begin{equation*}
  \begin{aligned}
    \langle\!\langle\bm x\rangle\!\rangle_\alpha^{S,g}
    & = \mu\langle\!\langle\bm x \rangle\!\rangle_{\alpha_j}^{S,g}
             +(1-\mu)\langle\!\langle\bm x\rangle\!\rangle_{\alpha_{j+1}}^{S,g}\\
    & = \frac{\alpha_{j+1}-\alpha}{(\alpha_{j+1}-\alpha_j)(1-\alpha)}\sum_{i=j+1}^n c_i|x|_{(i)}
            +\frac{\alpha-\alpha_j}{(\alpha_{j+1}-\alpha_j)(1-\alpha)}\sum_{i=j+2}^n c_i|x|_{(i)}\\
    & =  \frac{\alpha_{j+1}-\alpha}{(\alpha_{j+1}-\alpha_{j})(1-\alpha)}c_{j+1}|x|_{(j+1)}
          +\frac{1}{1-\alpha}\sum_{i=j+2}^n c_i|x|_{(i)}\\
    & = \frac{\alpha_{j+1}-\alpha}{1-\alpha}|x|_{(j+1)}+\frac{1}{1-\alpha}\sum_{i=j+2}^n (\alpha_i-\alpha_{i-1})|x|_{(i)}.
  \end{aligned}
\end{equation*}
When expressed in the framework of the OWL norm, the ordered weights corresponding to the scaled generalized-ES norm $\langle\!\langle\bm x\rangle\!\rangle_\alpha^{S,g}$ are as follows:
\begin{equation*}
    w_1=\cdots=w_j=0,\quad w_{j+1}=\frac{\alpha_{j+1}-\alpha}{1-\alpha},\quad w_{i}=\frac{\alpha_i-\alpha_{i-1}}{1-\alpha},\quad i=j+2,\ldots,n.
\end{equation*}
This construction ensures that $\langle\!\langle\bm x\rangle\!\rangle_\alpha^{S,g}$ inherits the structure of an OWL norm with a specific set of distortion-induced weights. Notably,
$$
     c_i=\alpha_{i}-\alpha_{i-1}=g\left(\frac{i}{n}\right)-g\left(\frac{i-1}{n}\right),\quad i=1,\ldots,n,
$$
are non-negative and increasing in $i$, which follows directly from the convexity of $g$. As a result, the corresponding weights $w_i$ satisfy the non-decreasing requirement of the OWL norm.

\begin{proposition}
 \label{pr-250510}
Let $g$ be a convex distortion function, and denote $\alpha_j=g(j/n)$ for $j=0, \ldots, n$.
\begin{itemize}
 \item[{\rm (i)}] For $\alpha\in[\alpha_{j},\alpha_{j+1}]$ with $ j=0,\ldots, n-2$, the dual norm of $\langle\!\langle\mathbf{x}\rangle\!\rangle_\alpha^{S,g}$  is given by
\begin{equation}
 \label{eq-250507}
       \langle\!\langle \bm x \rangle\!\rangle _\alpha^{S,g,\ast} = \max\left\{\|\bm x\|_1, \ \frac{1-\alpha}{1-g\left(\frac{i-1}{n}\right)} \sum_{k=i}^n |x|_{(k)}, \ i= j+2,\ldots, n\right\},\quad \bm x\in\R^n.
\end{equation}
 \item[{\rm (ii)}] For $\alpha\in (\alpha_{n-1}, 1)$, we have
\begin{equation}
 \label{eq-250508}
       \langle\!\langle \bm x \rangle\!\rangle _\alpha^{S,g,\ast} = \|\bm x\|_1,\quad \bm x\in\R^n.
\end{equation}
\end{itemize}
\end{proposition}

\begin{proof}
Note that
\begin{equation*}
   \tau_i = \left(\sum_{k=i}^n w_i\right)^{-1} = \begin{cases}    1, & i=1, \ldots, j+1,\\
            \frac{1-\alpha}{1-g\left(\frac{i-1}{n}\right)}, &  i=j+2, \ldots, n,
        \end{cases}
\end{equation*}
and that, for $i=1,\ldots, j+1$,
$$
   \tau_i\sum_{k=i}^n |x|_{(k)}=\sum_{k=i}^{n}|x|_{(k)}\leqslant \sum_{k=1}^n |x|_{(k)}=\tau_1\sum_{k=1}^n |x|_{(k)}= \|\bm x\|_1.
$$
Thus, the desired result follows from Lemma \ref{dual-lemma}.
\end{proof}

In particular, when the distortion function is $g_2(u)=u$, the dual norm of the scaled ES-norm $\langle\!\langle\bm x\rangle\!\rangle_\alpha^S$ is
\begin{equation}
   \label{dual-ES-norm}
    \langle\!\langle \bm x \rangle\!\rangle _\alpha^{S,\ast}
    = \max\left\{ \|\bm x \|_1,\ n(1-\alpha) \|\bm x\|_\infty\right\}
\end{equation}
for $\alpha\in \left[0, \frac{n-1}{n}\right]$. In fact, in this case, for $i=1, \ldots, n-1$,
$$
   \frac{1-\alpha}{1-g_2\left(\frac{i-1}{n}\right)} \sum_{k=i}^n |x|_{(k)} = \frac{n(1-\alpha)}{n-i+1} \sum_{k=i}^{n} |x|_{(k)} \leqslant n(1-\alpha)|x|_{(n)}= n(1-\alpha)||\mathbf{x}||_\infty.
$$
Thus, \eqref{dual-ES-norm} follows from \eqref{eq-250507}.
Therefore, the dual norm of the non-scaled ES-norm is
$$\langle\!\langle\mathbf{x}\rangle\!\rangle_\alpha^* = \max\left\{||\mathbf{x}||_\infty, \frac{||\mathbf{x}||_1}{n(1-\alpha)}\right\},$$
which is in concordance with Exercise IV.1.18 in \cite{B13} with $\alpha=\alpha_j$ for $j=0, \ldots, n-1$. The max operator in the dual norm expression (\ref{dual-ES-norm}) reveals two distinct regimes. First, when the vector entries are relatively balanced in magnitude, the dual norm coincides with the $L_1$-norm. Second, if the largest magnitude entry $|x|_{(n)}$ dominates all others, it solely determines the dual norm. This regime is particularly relevant in risk analysis, where a single catastrophic loss $|x|_{(n)}$ may exceed the aggregated impact of all other losses.

\section{Applications}
\label{sect-5}

\subsection{Projection problem}

We consider a projection problem under the scaled generalized-ES norm with the distortion function $g_1(u)=u^2$, and compare the objective values with those obtained using the scaled ES-norm. Let $P\subset\R^n$ be a convex polyhedron, and let $\bm w\notin P$. The goal to find a point $\bm w_P\in P$ that minimizes the distance between $\bm w$ and $\bm w_P$ induced by the norm $\langle\!\langle\cdot\rangle\!\rangle_\alpha^{S, g_1}$. This projection can be efficiently computed using convex or linear programming techniques.

Let $\bm w=(w_1, \ldots, w_n)\in \R^n$ and let the convex polyhedron $P$ be defined by a set of $m$ linear inequalities corresponding to $m$ hyperplanes. Each hyperplane $j$ is described by a vector $(a_1^j, \ldots, a_n^j, b^j)$, $j=1, \ldots, m$. Introducing the matrix $\bm A =(a_i^j)$ with $i=1,\ldots, n$ and $j=1,\ldots, m$, and the vector $\bm b^\top =(b^1, \ldots, b^m)$, the polyhedron $P$ is compactly expressed as the feasible set defined by the linear constraints: $\bm A \bm x \leqslant \bm b$, $\bm x \geqslant \bm 0$. The problem of projecting a point $\bm w$ onto $P$ under the scaled generalized-ES norm with $g_1$ is formulated as follows:
\begin{equation}\label{projection-primary}
    \begin{aligned}
    \min_{\bm x}\quad &\langle\!\langle\bm w-\bm x \rangle\!\rangle_\alpha^{S,g_1}\\
    \text{s.t.}\quad & \bm A \bm x \leqslant \bm b,\\
    &\quad \bm x \geqslant \bm 0.
    \end{aligned}
\end{equation}
By \eqref{eq-250506}, problem (\ref{projection-primary}) can be reformulated as:
\begin{equation}\label{order-constraint}
    \begin{aligned}
    \min_{x_1,\ldots,x_n,t}\quad & t+\frac{1}{1-\alpha}\sum_{i=1}^{n} c_{i}z_{i}\\
    \text{s.t.}\quad & z_{i}\geqslant |x-w|_{(i)} -t,\quad i=1,\ldots,n,\\
    &z_{i}\geqslant0,\ x_i\geqslant0, \quad i=1,\ldots,n,\\
    &\sum_{i=1}^n a_{i}^j  x_i\leqslant b^j, \quad j=1,\ldots, m.
    \end{aligned}
\end{equation}
The challenge lies in the order-statistic constraint in \eqref{order-constraint}, which is nonlinear and non-convex due to the sorting operation. To address this, we employ a mixed-integer linear programming (MILP) approach by introducing a binary matrix $\bm U = (u_{ik})\in\{0,1\}^{n\times n}$ that explicitly models the ordering of $|\bm x-\bm w|$. The absolute value term is linearized via auxiliary variables $\bm y= (y_1,\ldots, y_n)$ such that $y_i\geqslant x_i-w_i$ and $y_i\geqslant w_i-x_i$ for each $i$. And let $\bm s =(s_1, \ldots, s_n)$ be the order statistics of $\bm y$. The projection problem is then reformulated as the following mixed-integer linear problem:
\begin{equation}\label{milp}
 \begin{aligned}
   \min_{\bm x, \bm y, \bm s, \bm z, t, \bm U} \quad & t + \frac{1}{1 - \alpha} \sum_{i=1}^n c_i z_i \\
   \text{s.t.} \quad
      & \sum_{i=1}^n a_{i}^j  x_i\leqslant b^j, \quad j=1,\ldots,m, \\
      & x_i \geqslant 0, z_i \geqslant s_i - t,\quad z_i \geqslant 0, \quad i=1,\ldots,n,\\
      & y_i \geqslant x_i - w_i,\quad y_i \geqslant w_i - x_i, \quad i = 1,\ldots,n, \\
      & \sum_{k=1}^{n} u_{ik} = 1, \quad i = 1,\ldots,n, \\
      &\sum_{i=1}^{n} u_{ik} = 1, \quad  k = 1,\ldots,n, \\
      & s_i = \sum_{k=1}^n u_{ik} y_k, \quad  i = 1,\ldots,n, \\
      & s_i \leqslant s_{i+1}, \quad  i = 1,\ldots,n-1, \\
      & u_{ik} \in \{0, 1\}, \quad  i,k = 1,\ldots,n.
  \end{aligned}
\end{equation}

For the numerical experiment, we randomly generate the matrix $\bm A\in\R^{m\times n}$ and the vector $\bm b\in\R^m$. Specifically, the entries of $\bm A$ and a reference vector $\bm x^\ast\in \R^m$ are independently drawn from the uniform distribution on $[0, 1]$. We then construct $\bm b=\bm A\bm x^\ast +1/10$ to ensure that the feasible polyhedron $P =\{\bm A\bm x\geqslant \bm b, \, \bm x \geqslant \bm 0\}$ is non-empty. To generate a query point $\bm w\notin P$, we set $\bm w = \bm x^\ast + \bm d/2$, where $\bm d\in\R^{m\times 1}$ is another vector with its entries sampled from uniform distribution on $[0,1]$. This setup basically guarantees that $\bm w$ lies outside the feasible region, making the projection problem non-trivial.

The integer matrix $\bm U$ in \eqref{milp} involves binary variables to model the sorting operation, leading to a mixed-integer program with $O(n^2)$ complexity. As $n$ increases, the number of variables and constraints becomes computationally prohibitive, even for modern solvers like \texttt{Gurobi}. Therefore, the method is suitable only for small to moderate dimensions (e.g., $n\leqslant 30$), beyond which solving time grows rapidly and may become impractical. Therefore, we deal with the case with $n=10$ variables and $m=5$ constraints and employ \texttt{YAMLIP} \cite{L04} and solver \texttt{Gurobi} to solve (\ref{milp}). Figure \ref{obj-val compare} compares the objective values obtained from projection problems using the scaled ES-norm \citep{PU14} versus the scaled generalized ES-norm under varying levels of $\alpha$. Figure \ref{obj-val compare} reveals that the projection problem using the scaled generalized-ES norm yields slightly higher objective values compared to those based on the scaled ES norm, both of which show monotonically increasing trends with respect to $\alpha$. Notably, the two objective values coincide when $\alpha=0.9$, since the generalized-ES norm and the ES norm both degenerate to the $n$th-order statistic.

\begin{figure}
   \centering
    \includegraphics[scale=0.7]{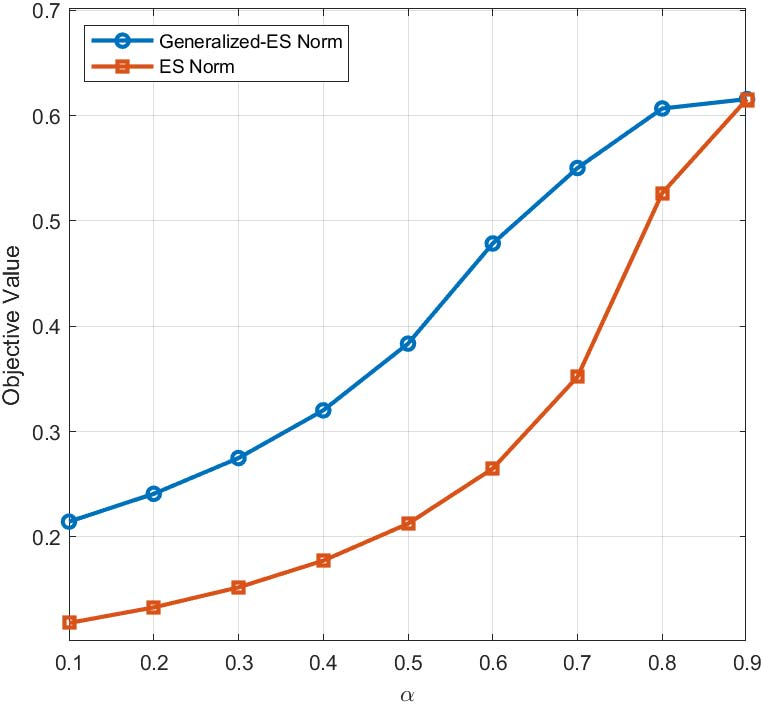}
    \caption{Objective values comparison between projection problems using the scaled ES norm and the scaled  generalized-ES norm.}
    \label{obj-val compare}
\end{figure}

\subsection{Anomaly detection}
\label{sect-5.2}

A point anomaly in time series data is defined as a point that deviates from the normal behavior of data. Detecting anomalies in financial time series is a crucial task for risk management and the identification of extreme market behavior. Financial data, especially returns from assets, often exhibit heavy-tailed and non-Gaussian characteristics. This is particularly true for cryptocurrencies like Bitcoin, which are known for their high volatility and frequent large deviations.

To detect point anomalies, traditional techniques such as the modified $Z$-score based on the median absolute deviation (MAD), the peaks over threshold (POT) method in extreme value theory (EVT), and machine learning methods like isolation forest have been widely used for this purpose. For comprehensive state-of-the-art on anomaly detection techniques in time series data, see \cite{BCML21} and \cite{SAL21}. However, these approaches may be overly sensitive to distributional assumptions or lack interpretability when applied to data with extreme tail behavior. In this section, we propose a novel point anomaly detection approach based on the scaled generalized-ES norm, which does not depend on explicit distributional assumptions or static thresholds and thus is adaptable to evolving market dynamics.

We adopt a rolling-window framework over the time series of interest. For each time point $t$, we exact a fix-sized window of recent observations $\left\{r_{t-d+1},\ldots, r_t\right\}$, where $d$ denotes the window size. The first $d-1$ values in the window are used to compute the scaled generalized-ES norm as the reference threshold to identify whether $r_t$ is a point anomaly. Specifically, for each $t$, we calculate the scaled generalized-ES norm
\begin{equation*}
   I_t= \langle\!\langle \mathbf{r}_{t-d+1:t-1} \rangle\!\rangle_\alpha^{S,g},
\end{equation*}
where $\mathbf{r}_{t-d+1:t-1}=(r_{t-d+1}, \ldots, r_{t-1})$. If the absolute value of $r_t$ exceeds $I_t$, then $r_t$ is regarded as a point outlier. This procedure is repeated across all time points using a moving window, yielding a binary indicator series that marks the presence or absence of outliers.

Since the scaled generalized-ES norm in this framework flexibly incorporates a distortion function $g$ and a tunable tail-sensitivity parameter $\alpha$,  it allows for selective emphasis on extreme deviations, thereby enabling tailored detection based on the characteristics of the data or the preference of the analyst. As a result, our approach is particularly well-suited for high-frequency, heavy-tailed financial data, and offers greater interpretability compared to machine learning-based methods such as isolation forest.

We apply this framework to daily return data of Bitcoin over the period from 2020 to 2024. First, we examine the impact of choosing a quadratic distortion function $g_1(x)=x^2$ versus no distortion (i.e. $g_2(x)=x$) under three levels of $\alpha=0.9, 0.95$ and 0.99 respectively in order to address the importance of generalizing the ES-norm with convex distortion. A rolling window of size 30 is adopted, corresponding approximately to one month of returns. As shown in figure \ref{g-alpha}, introducing a convex distortion allows for more effective identification of truly extreme return events, while the scaled ES-norm tends to mark a broader range of observations, including some that may not reflect genuine anomalies. Furthermore, as $\alpha$ approaches 1, the distinction between the scaled generalized-ES norm and the scaled ES-norm diminishes as both settings increasingly concentrate on the most extreme observations in the tail.
\begin{figure}[htbp]
  \centering
 \subfigure[\ ]{ \includegraphics[scale=0.4]{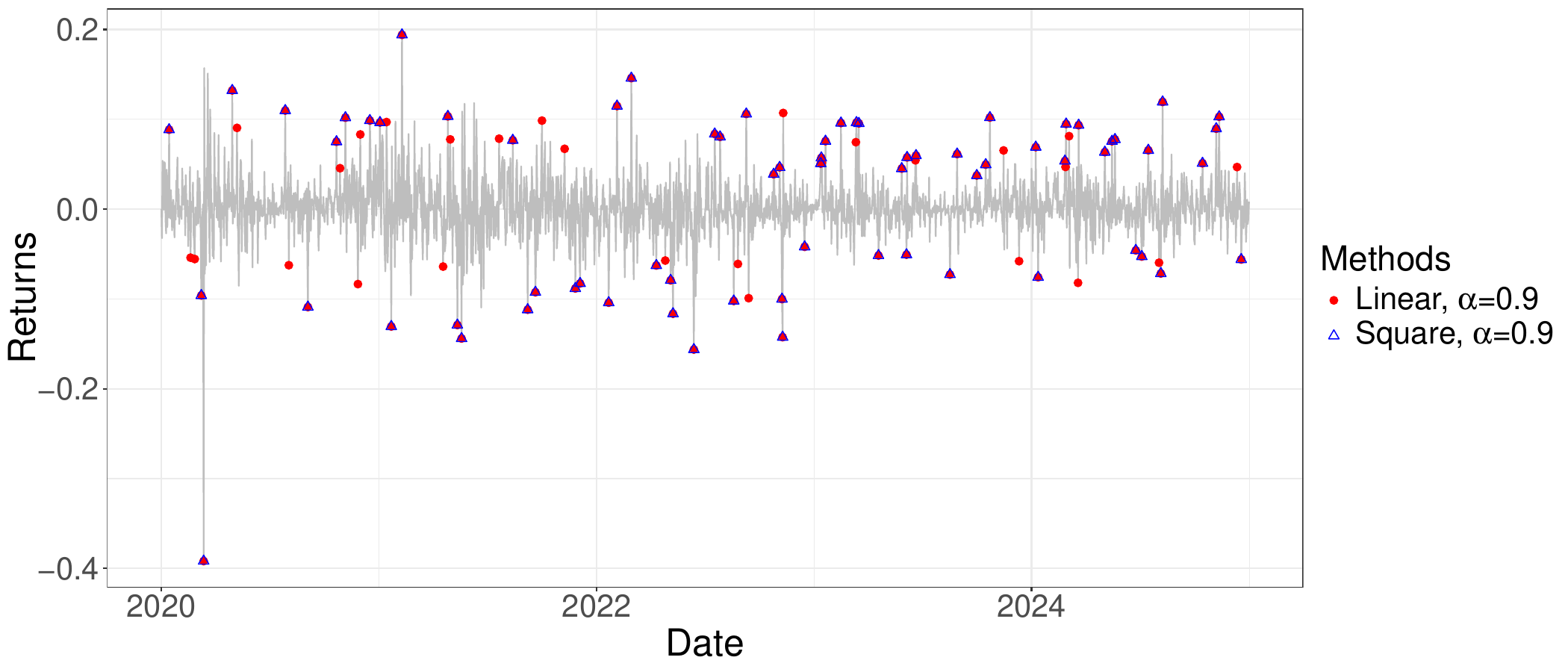} }
  \subfigure[\ ]{ \includegraphics[scale=0.4]{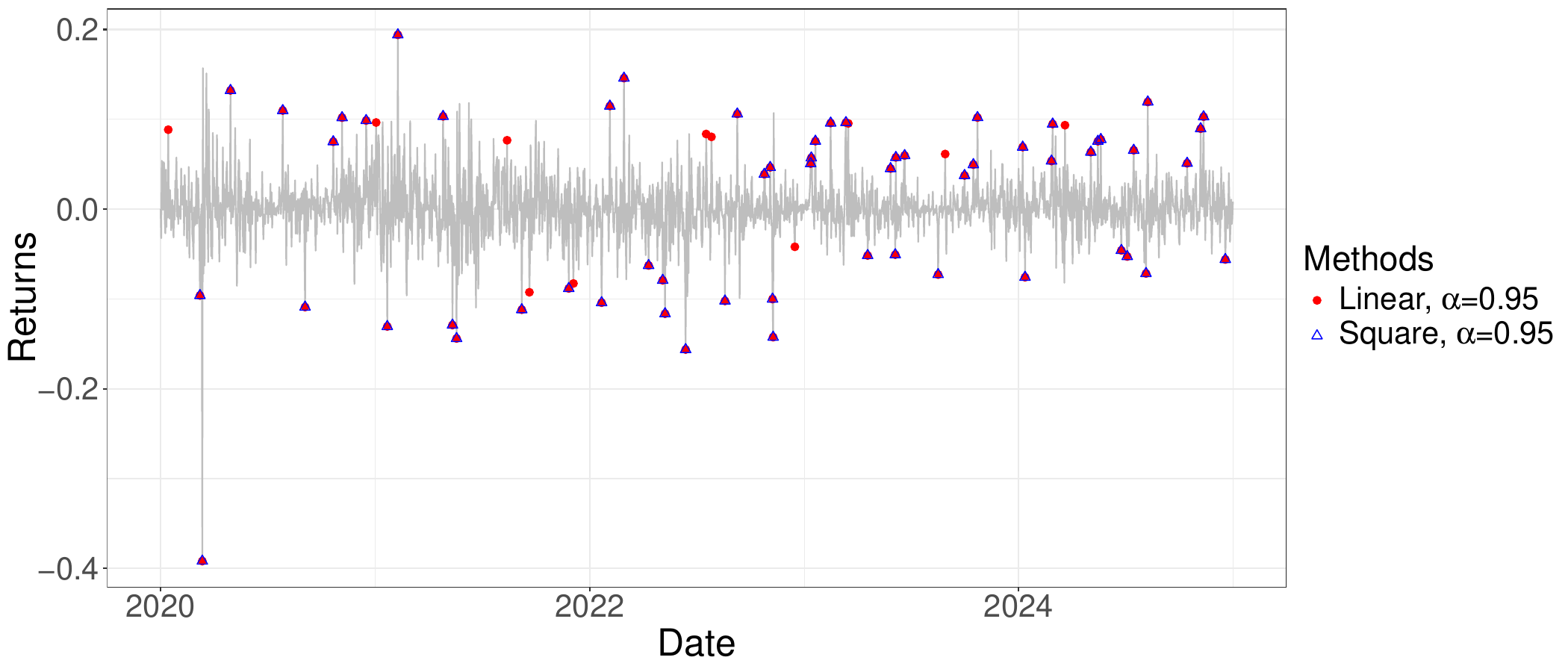} }
  \subfigure[\ ]{ \includegraphics[scale=0.4]{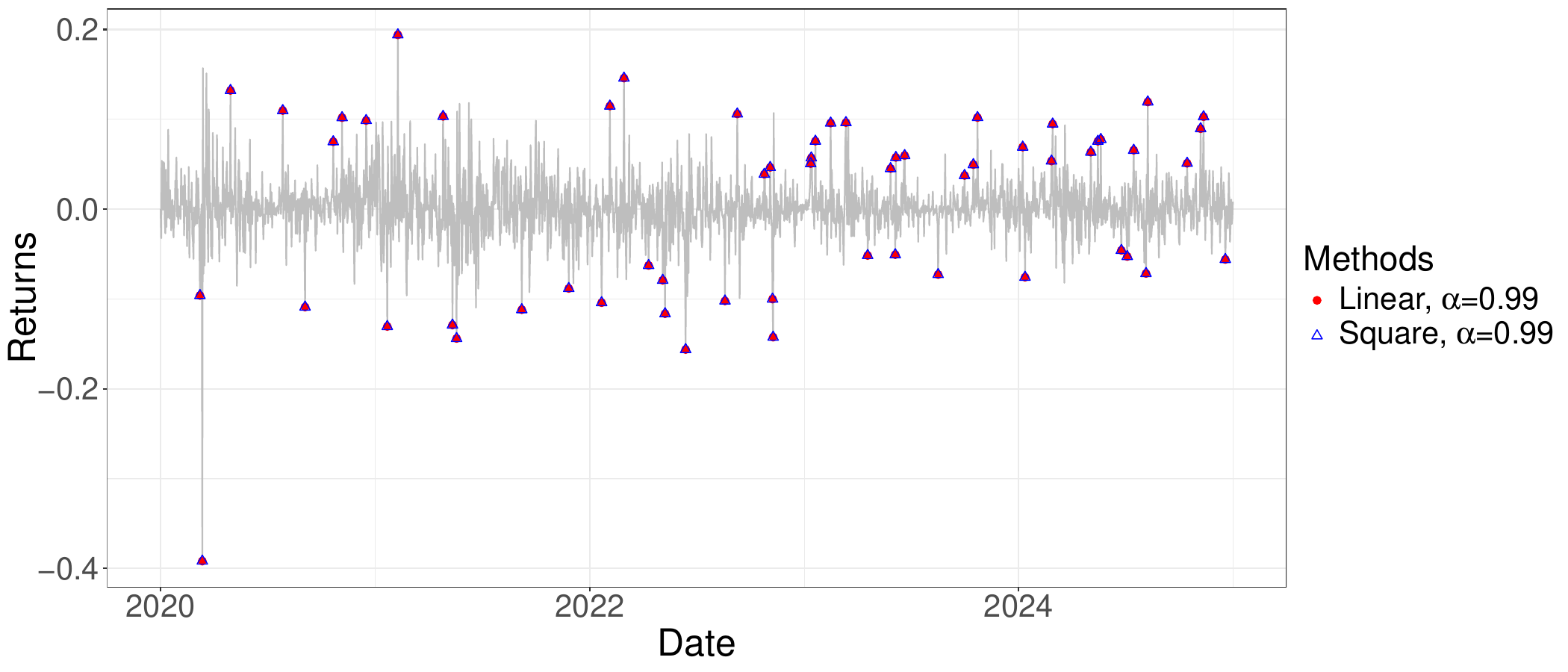} }
  \caption{Outlier detection results for $g_1(x)=x^2$ and $g_2(x)=x$ under three different levels $\alpha=0.9, 0.95, 0.99$.}
   \label{g-alpha}
\end{figure}

Next, we compare our approach with other traditional outlier detection approaches that are suitable for heavy-tailed financial time series data.
\begin{itemize}
  \item Introduced by \cite{IH93} and discussed in \cite{LLKBL13}, the modified $Z$-score based on MAD is calculated as
    \begin{equation*}
        \text{Modified $Z$-score} = \frac{0.6745\times (r_t-\text{median}(\mathbf{r}_{t-d+1:t-1}))}{\text{MAD}(\mathbf{r}_{t-d+1:t-1})},
    \end{equation*}
    where $\text{MAD}(\mathbf{r}_{t-d+1:t-1})=\text{median}(|\mathbf{r}_{t-d+1:t-1}-\text{median}(\mathbf{r}_{t-d+1:t-1})|)$, and the scaling factor 0.6745 ensures that the modified $Z$-score is consistent with the standard $Z$-score under a Gaussian distribution. It marks observations exceeding a certain threshold as point outliers. For heavy-tailed Bitcoin returns, we set the threshold of the absolute modified $Z$-score to 3. Any observation exceeding this threshold is flagged as a point outlier. The window width is set to be 30, same as our norm-based approach.

 \item The Peaks Over Threshold (POT) method, a key approach in Extreme Value Theory (EVT), models the tail behavior of financial returns to identify rare event (see \citealp{CBTD01,SM12}). The POT method uses a generalized Pareto distribution (GPD) to model exceedances above a high threshold $u$:
    \begin{equation*}
        G_{\xi,\beta}(x)=1 - \left(1+\xi\frac{x-u}{\beta}\right)^{-1/\xi}, \quad x>u,
    \end{equation*}
    where $\xi$ and $\beta$ are shape and scale parameters, respectively. Observations beyond some quantile threshold $v$ of data in the window are classified as outliers. The POD method requires an adequately large observation window to ensure sufficient exceedance date, as sparse exceedances can result in high parameter estimation variance and consequently lead to instability in GPD fitting. In our study, we adopt a fixed window width of 180 observations, and consider outliers both in the right tail and the left tail in the window to examine both positive and negative extremes. To deal with the right tail, we set the threshold $u$ at the 90th percentile of the windowed data to ensure sufficient tail data for parameter estimation. To examine extreme outliers, we set the quantile threshold $v$ as the 99th percentile, effectively detecting the most deviant points. The outliers in left tail are handled analogously.

 \item Isolation Forest is an efficient unsupervised learning algorithm specifically designed for anomaly detection in high-dimensional datasets, as proposed in  \cite{LTZ08}. The core idea is based on the observation that anomalies are data points that are ``few and different" and therefore they tend to be more easily separable from the rest of the data. Isolation Forest isolates anomalies explicitly by recursively partitioning the data space. The algorithm works by constructing an ensemble of binary trees, known as isolation trees (iTrees). Each tree is built by randomly selecting a feature and then randomly choosing a split value within the range of that feature. This process is repeated recursively until either the point is isolated (i.e., the partition contains only one sample) or a maximum tree height is reached.

    For each observation $x$, its path length $h(x)$ is defined as the number of edges from the root node to the leaf node where $x$ is isolated. Anomalies are expected to have shorter average path lengths since they can be isolated early due to their rarity and feature-space distinctiveness.
    The anomaly score for a point $x$ is computed as
    \begin{equation*}
        s(x, n)=2^{-\E [h(x)]/c(n)},
    \end{equation*}
    where $\E [h(x)]$ is the average path length across all trees in the ensemble, $n$ is the subsample size used to build each tree and $c(n)$ is a normalization factor representing the average path length of unsuccessful searches in a binary search tree, which corrects for the bias introduced by different sample sizes.
    The expression for $c(n)$ is approximated by
    \begin{equation*}
        c(n)=2H(n-1)-\frac{2(n-1)}{n},
    \end{equation*}
    where $H(i)$ is the $i$-th harmonic number $H(i)=ln(i)+\gamma$, and $\gamma\approx0.5772$ is the Euler–Mascheroni constant. This approximation holds for $n>2$. The resulting anomaly score $s(n)$ lies in the interval $(0,1)$, with values closer to 1 indicating a higher likelihood of being an anomaly. Intuitively, if a point can be quickly isolated with a small $h(x)$, it is considered more anomalous.

    In our implementation, we apply Isolation Forest using a moving window of width 180, meaning that at each step, the algorithm is applied to the most recent 180 observations. We compute the anomaly score $s(x)$ for each point in the window and label points with scores exceeding 0.7 as anomalies. This threshold is selected to balance sensitivity and false detection rates, and can be tuned depending on the tolerance to false positives.
\end{itemize}

The results in Figure \ref{anomaly comparison} demonstrate that our norm-based method successfully captures most common outliers identified by other approaches, confirming its effectiveness. Table \ref{tab-no.outliers} presents a cross-method comparison of outlier detection results, where GES-square and GES-linear are generalized-ES norm-based methods with distortion functions $g_1$ and $g_2$, respectively. The diagonal cells show each method's total outlier counts, while off-diagonal cells indicate shared detections between method pairs. Parenthetical percentages reflect the overlap rate relative to each column method's total (e.g., 84.06\% means 84.06\% of that outliers detected by GES-linear were also detected by GES-square). Compared to alternative methods, our approach offers several key advantages. It provides greater flexibility than the modified $Z$-score based on MAD through tunable distortion function $g$ and level $\alpha$, allowing for customized detection. Moreover, the proposed method requires significantly smaller window sizes compared to both EVT and isolation forest approaches. This computational advantage proves particularly valuable for Bitcoin time series analysis, where frequent regime shifts necessitate adaptive modeling frameworks. Notably, the EVT method demonstrates a relatively high false positive rate even with a $99$th percentile quantile threshold $(v=0.99)$, suggesting that the current window width of $180$ observations may be insufficient for robust parameter estimation. Unlike black-box machine learning methods (e.g., isolation forest), our approach derives directly from risk measure theory, offering mathematically tractable interpretations that are essential for financial risk management.

\begin{figure}
    \centering
    \includegraphics[scale=0.45]{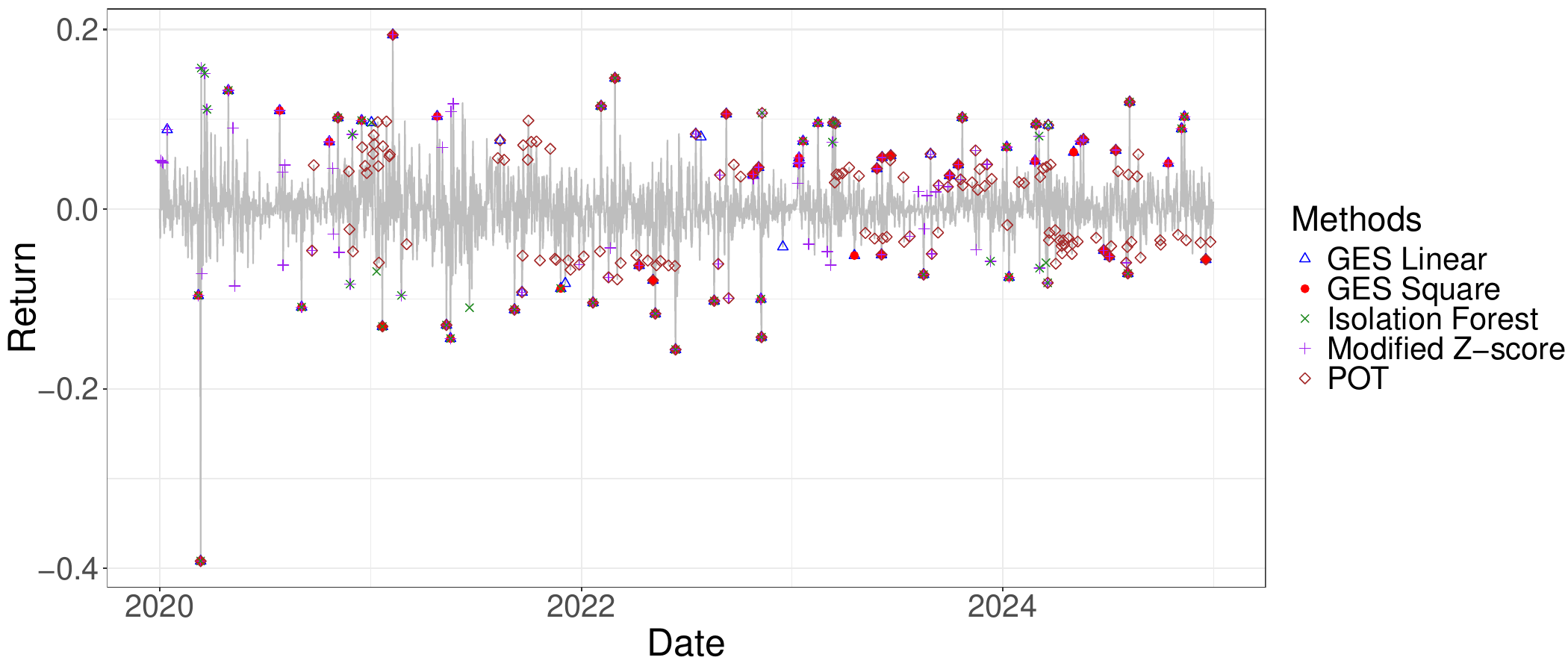}
    \caption{Outliers detected by different methods}
    \label{anomaly comparison}
\end{figure}

\begin{table}[htbp]
  \centering
  \caption{The diagonal entries show the total number of outliers detected by each method. The off-diagonal entries indicate the counts of common outliers detected by each pair of methods. The percentages in parentheses represent the proportion of these common outliers relative to the total outliers detected by the column method.}
  \label{tab-no.outliers}  \medskip
  \begin{tabular}{c cc cc c}
    \toprule
          & \multicolumn{1}{c}{GES-linear} & \multicolumn{1}{c}{GES-square} & \multicolumn{1}{c}{Isolation Forest} & \multicolumn{1}{c}{Modified $Z$-score} & \multicolumn{1}{c}{POT} \\
    \midrule
    GES-linear           & \textbf{69}           & 58 (100.00\%) & 34 (69.39\%) & 56 (50.91\%) & 41 (26.28\%) \\
    GES-square           & 58 (84.06\%) & \textbf{58}           & 31 (63.27\%) & 51 (46.36\%) & 35 (22.44\%) \\
    Isolation Forest     & 34 (49.28\%) & 31 (53.44\%) & \textbf{49}      & 42 (38.18\%) & 22 (14.10\%) \\
    Modified Z-score     & 56 (81.16\%) & 51 (87.93\%) & 42 (85.71\%) & \textbf{110}     & 51 (32.69\%) \\
    POT                  & 41 (59.42\%) & 35 (60.34\%) & 22 (44.90\%) & 51 (46.36\%) & \textbf{156}  \\
    \bottomrule
  \end{tabular}

\end{table}

To further emphasize the effectiveness of our outlier detection method, we again employ this framework to the Chicago board options exchange volatility index (abbreviated as VIX). The VIX, often termed the $``$fear gauge", is a real-time market index representing the stock markets expectations of 30-day forward-looking volatility, derived from S$\&$P 500 index options. As a benchmark for investor sentiment and market risk, the VIX tends to spike during periods of financial stress, making it a critical tool for assessing systemic instability; see \cite{Wha00}.

We set the distortion function to be $g_1(x)=x^2$ and let $\alpha=0.95$. As figure \ref{anomalies in VIX} shows, we identified several anomalies in the daily highs of the VIX between 2018 and 2024, demonstrating that our method is capable of identifying known episodes of market disruption, often ahead of or in alignment with major macro-financial events. The key findings are as follows:
\begin{itemize}
    \item February 2018 flash crash: While the U.S. equity market experienced a rapid sell-off on February 5th-8th, our method flagged anomalous VIX behavior as early as February 2nd, suggesting mounting pressure in the options market prior to the event.

    \item March 2020 COVID-19 meltdown: The U.S. stock market collapses multiple times in March 2020 due to panic over the rapid global spread of COVID-19. Although VIX peaked during the week of March 9th, our anomaly detection indicates that VIX began showing sustained abnormal spikes starting from February 24th, 2020.

    \item Tapering and geopolitical risks in 2021-2022: Between late 2021 and March 2022, our method detected five intermittent anomalous VIX spikes. This period corresponds to the Federal Reserves tightening signals as well as the escalation of Russia-Ukraine conflict, which added significant geopolitical risk to market sentiment.

    \item August 2024 volatility surge: On August 5th, 2024, the VIX approached levels not seen since the COVID crash and recorded the largest single-day intraday swing in its history. Our method identified August 2nd, 5th, and 6th as anomalous, once again suggesting the models capacity to detect extreme volatility with both accuracy and timeliness.
\end{itemize}

The anomalies consistently lead or coincide with independently verified crises, corroborating their validity. Early warnings (e.g., February 2018 and February 2020) imply that options markets may price in latent risks before their realization in cash equities hypothesis supported by the VIXs derivative-based construction. Moreover, the 2021-2022 and 2024 detections demonstrate the methods robustness across diverse shock types (monetary, geopolitical, and liquidity-driven). By identifying volatility outliers with temporal precision, the method demonstrates significant promise for real-time financial surveillance and systemic risk monitoring.

\begin{figure}
    \centering
    \includegraphics[scale=0.4]{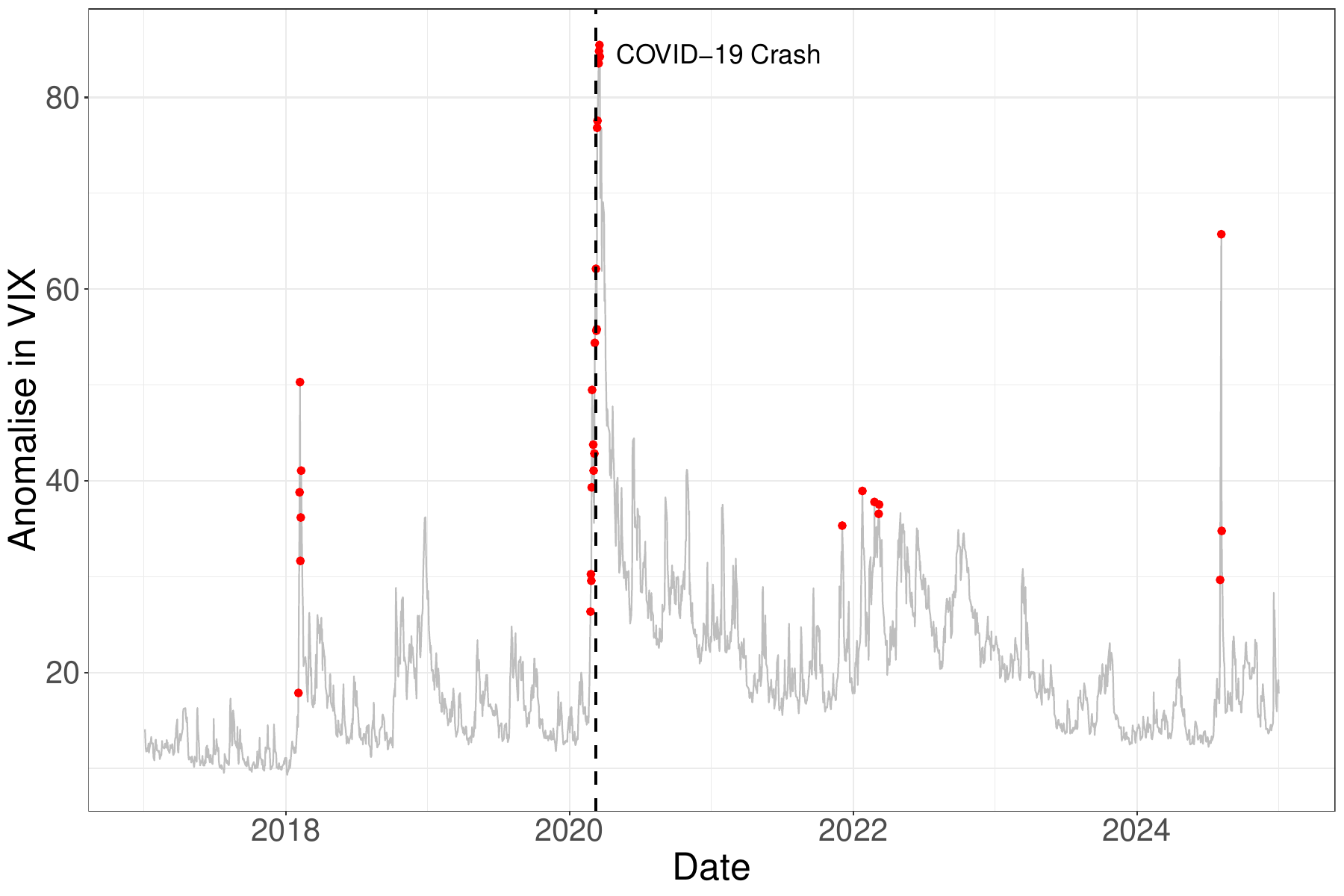}
    \caption{Outliers in VIX from 2018 to 2024}
    \label{anomalies in VIX}
\end{figure}

\section*{Funding}

T. Hu would like to acknowledge financial support from National Natural Science Foundation of China (No. 72332007, 12371476). Z. Zou is supported by National Natural Science Foundation of China (No. 12401625), the China Postdoctoral Science Foundation (No. 2024M753074), and the Fundamental Research Funds for the Central Universities, China (WK2040000108).

\section*{Disclosure statement}

No potential conflict of interest was reported by the authors.



\begin{thebibliography}{10}

\bibitem[\protect\citeauthoryear{Bertsimas et al.}{2004}]{BPS04}
   Bertsimas, D., Pachamanova, D. and Sim, M. (2004). Robust linear optimization under general norms. \emph{Operations  Research Letters}, \textbf{32}(6), 510-516.

\bibitem[\protect\citeauthoryear{Bhatia}{2013}]{B13}
   Bhatia, R. (2013). \emph{Matrix Analysis}. Volume 169. Springer Science \& Business Media.

\bibitem[\protect\citeauthoryear{Bl\'{a}zquez-Garc\'{i}a et al.}{2021}]{BCML21}
   Bl\'{a}zquez-Garc\'{i}a, A., Conde, A., Mori, U. and Lozano, J.A. (2021). A review on outlier/anomaly detection in time series data. \emph{ACM Computing Surveys}, \textbf{54}(3), Article 56, 33 pages.

\bibitem[\protect\citeauthoryear{Coles et al.}{2001}]{CBTD01}
   Coles, S., Bawa, J., Trenner, L. and Dorazio, P. (2001). \emph{An Introduction to Statistical Modeling of Extreme Values}. Springer, London.

\bibitem[\protect\citeauthoryear{Deza and Deza}{2016}]{DD16}
   Deza, M.M. and Deza, E. (2016). \emph{Encyclopedia of Distances}, Fourth Edition. Springer.

\bibitem[\protect\citeauthoryear{Dhaene et al.}{2012}]{DKLT12}
   Dhaene, J. and Kukush, A., Linders, D. and Tang, Q. (2012). Remarks on quantiles and distortion risk measures. \emph{European Actuarial Journal}, \textbf{2}(2), 319-328.

\bibitem[\protect\citeauthoryear{Gong et al.}{2024}]{GZGH24}
   Gong, S., Zou, Z., Guan, M. and Hu, T. (2024). Generalized Expected-Shortfalls based on distortion risk measures. Submitted to {\em Insurance: Mathematics and Economics}.

\bibitem[\protect\citeauthoryear{Iglewicz and Hoaglin}{1993}]{IH93}
  Iglewicz, B. and Hoaglin, D. C. (1993). \emph{How to Detect and Handle Outliers}. ASQC/Quality Press.

\bibitem[\protect\citeauthoryear{Kremer et al.}{2020}]{KLBP20}
  Kremer, P.J., Lee, S., Bogdan, M. and Paterlini, S. (2020). Sparse portfolio selection via the sorted $\ell_1$-norm. \emph{Journal of Banking and Finance}, \textbf{110}, 105687.

\bibitem[\protect\citeauthoryear{Leys et al.}{2013}]{LLKBL13}
   Leys, C., Ley, C., Klein, O., Bernard, F. and Licata, L. (2013). Detecting outliers: Do not use standard deviation around the mean, use absolute deviation around the median, \emph{Journal of Experimental Social Psychology}, \textbf{49}(4), 764-766.

\bibitem[\protect\citeauthoryear{Li and Li}{2022}]{LL22}
   Li, Q. and Li, X. (2022). Fast projection onto the ordered weighted $\ell_1$ norm ball. \emph{Science China: Mathematics}, \textbf{65}, 869-886.

\bibitem[\protect\citeauthoryear{Liu et al.}{2008}]{LTZ08}
   Liu, F. T., Ting, K. M. and Zhou, Z. (2008). Isolation forest. \emph{2008 Eighth IEEE International Conference on Data Mining}, Pisa, Italy, pp. 413-422.

\bibitem[\protect\citeauthoryear{L{\"{o}}fberg}{2004}]{L04}
   L{\"{o}}fberg, J. (2004). YALMIP: A toolbox for modeling and optimization in MATLAB.
   \emph{2004 IEEE International Conference on Robotics and Automation}, Taipei, Taiwan, pp. 284-289.

\bibitem[\protect\citeauthoryear{Mazza-Anthony et al.}{2021}]{MMC21}
   Mazza-Anthony, C., Mazoure, B. and Coates, M. (2021). Learning gaussian graphical models with ordered weighted $\ell_1$ regularization. \emph{IEEE Transactions on Signal Processing}, \textbf{69}, 489-499.

\bibitem[\protect\citeauthoryear{Pavlikov and Uryasev}{2014}]{PU14}
   Pavlikov, K. and Uryasev, S. (2014). CVaR norm and applications in optimization. \emph{Optimization Letters}, \textbf{8}(7), 1999-2020.

\bibitem[\protect\citeauthoryear{Pichler}{2024}]{Pic24}
  Pichler, A. (2024). Connection between higher order measures of risk and stochastic dominance. \emph{Computational Management Science}, \textbf{21}, paper number: 41 (28 pages).

\bibitem[\protect\citeauthoryear{Rockafellar and Uryasev}{2002}]{RU02}
  Rockafellar, R.T. and Uryasev, S. (2002). Conditional value-at-risk for general loss distributions. \emph{Journal of Banking and Finance}, \textbf{26}(7), 1443-1471.

\bibitem[\protect\citeauthoryear{Scarrote and MacDonald}{2012}]{SM12}
   Scarrott, C. and MacDonald, A. (2012). A review of extreme value threshold estimation and uncertainty quantification. \emph{REVSTAT-Statistical journal}, \textbf{10}(1), 33-60.

\bibitem[\protect\citeauthoryear{Shaukat et al.}{2021}]{SAL21}
  Shaukat, K., Alam, T. M., Luo, S., et al. (2021). A review of time-series anomaly detection techniques: A step to future perspectives. In: Arai, K. (eds) \emph{Advances in Information and Communication}, Advances in Intelligent Systems and Computing, vol. 1363, pp. 865-877.

\bibitem[\protect\citeauthoryear{Wang et al.}{2020}]{WWW20}
  Wang, R., Wei, Y. and Willmot, G.E. (2020). Characterization, robustness, and aggregation of signed Choquet integrals. \emph{Mathematics of Operations Research}, \textbf{45}(3), 993-1015.

\bibitem[\protect\citeauthoryear{Whaley}{2000}]{Wha00}
  Whaley, R. E. (2000). The investor fear gauge. \emph{Journal of Portfolio Management}, \textbf{26}(3), 12-17.

\bibitem[\protect\citeauthoryear{Zeng and Figueiredo}{2014}]{ZF14}
   Zeng, X. and Figueiredo M. (2014). Decreasing weighted sorted $\ell_1$ regularization. \emph{IEEE Signal Processing Letters}, \textbf{21}, 1240-1244.

\bibitem[\protect\citeauthoryear{Zeng and Figueiredo}{2015}]{ZF15}
   Zeng, X. and Figueiredo M. (2015). The ordered weighted $\ell_1$ norm: Atomic formulation, projections, and algorithms. \emph{arXiv:1409.4271v5}.




\end{thebibliography}
\end{document}